\documentclass[journal]{IEEEtran}
\usepackage[font=footnotesize,labelfont=bf]{caption}
\usepackage{amsmath}
\usepackage{booktabs}
\usepackage{cite}

\usepackage[inline]{enumitem}
\usepackage[export]{adjustbox}
\usepackage{wrapfig}
\usepackage{subfigure}

\usepackage{amssymb,amsfonts} 
\usepackage{amsthm} 
\usepackage{commath} 
\usepackage{algorithmic}
\usepackage{mathtools, nccmath}

\usepackage{url}
\usepackage{array}
\usepackage{textcomp}
\usepackage[inline]{enumitem}
\usepackage{xcolor}
\usepackage{comment}  
\usepackage[final]{pdfpages}

\usepackage{setspace}

\newcommand{\dosint}[2]{%
  \ifx#1\displaystyle
    \displaysint
  \else
    \normalsint{#1}%
  \fi
}
\newcommand{\displaysint}{\displaystyle\mathsf{s}\mkern-18mu}
\newcommand{\normalsint}[1]{%
  \smallers{#1}\ifx#1\textstyle\mkern-9mu\else\mkern-8.2mu\fi
}
\newcommand{\smallers}[1]{%
  \vcenter{\hbox{$\ifx#1\textstyle\scriptstyle\else\scriptscriptstyle\fi\mathsf{s}$}}%
}

\graphicspath{{./figures/}}

\let\OLDthebibliography\thebibliography
\renewcommand\thebibliography[1]{
  \OLDthebibliography{#1}
  \setlength{\parskip}{0pt}
  \setlength{\itemsep}{0pt plus 0.3ex}
}

\begin{document}
\theoremstyle{definition}

\newtheorem{theorem}{Theorem}

\theoremstyle{definition}
\newtheorem{definition}{Definition}

\theoremstyle{definition}
\newtheorem{corollary}{Corollary}

\theoremstyle{definition}
\newtheorem{proposition}{Proposition}

\theoremstyle{definition}
\newtheorem{lemma}{Lemma}

\theoremstyle{definition}
\newtheorem{claim}{Claim}

\theoremstyle{definition}
\newtheorem{conjecture}{Conjecture}

\theoremstyle{definition}
\newtheorem{observation}{Observation}

\title{Fast-Response \\ Variable Frequency DC-DC Converters Using\\ Switching Cycle Event-Driven Digital Control}
\date{}
\author{Xiaofan~Cui,~\IEEEmembership{Student Member,~IEEE,}
        and~Al-Thaddeus~Avestruz,~\IEEEmembership{Senior Member,~IEEE}
\thanks{
The discussion and data are presented in part at the IEEE Energy Conversion Congress and Exposition, Baltimore, MD, USA, September 2019, in part by the 19th Workshop on Control and Modeling for Power Electronics, Padua, Italy, June 2018, and in part by the 2019 American Control Conference, Philadelphia, PA, USA, July 2019. This article presents new advances in theory and analysis.

The authors are with the Department of Electrical and Computer Engineering, University of Michigan, Ann Arbor,
MI, 48109 USA (e-mail: cuixf@umich.edu; avestruz@umich.edu).
}
\thanks{Manuscript received xx, 2021; revised xx, 2021.}}

\markboth{ \hspace{1in}  {\bf MANUSCRIPT}}%
{Shell \MakeLowercase{\textit{et al.}}: Bare Demo of IEEEtran.cls for IEEE Journals}

\maketitle
\thispagestyle{headings}
\pagestyle{headings}

\renewcommand{\figurename}{Fig.}
\begin{abstract}
This paper investigates a new method to model and control variable-frequency power converters in a switching-synchronized sampled-state space for cycle-by-cycle digital control. 
There are a number of significant benefits in comparison to other methods including fast dynamic performance together with ease of design and implementation. 
Theoretical results are presented and verified through hardware, and simulations of a current-mode buck converter with constant on-time and a current-mode boost converter with constant off-time. 
Dynamic voltage scaling for microprocessors and LiDAR are among the applications that can benefit.
\end{abstract}

\section{Introduction}    
\label{sec:Intro}
\IEEEPARstart{E}{fficiency} is achieved by power converters that provide the exact energy at the exact time. For example, in an electronic system whose power demand fluctuates rapidly, dynamic voltage scaling (DVS) is a crucial and widely used technique to optimize energy efficiency. However, traditional controllers are generally not able to provide both high speed dynamic response and programmable flexibility at the same time: for example, analog controllers lack the flexibility in speedily tuning controller parameters for varying output voltage levels.
Among these, constant\nobreakdash-frequency dc\nobreakdash-dc converters cannot be controlled cycle\nobreakdash-by\nobreakdash-cycle in a simple way over a wide voltage conversion ratio. Even existing variable frequency power converters with digital controllers have extraordinary complexities in architecture, hardware, algorithms, or sensitivities to parameter variations. A new digital control framework is needed to overcome these pervasive limitations.

There are a number of loads that are dynamically demanding.  Processors require higher voltages during intensive tasks \cite{peng2013instruction}.  Memories  are able to work at a lower voltage supply when read/write bandwidth is less \cite{david2011memory}. Wireless network modules choose working voltages based on communication channel conditions and packet throughput \cite{homchaudhuri2016dynamic}. 
LiDAR (Light Detection and Ranging) can dynamically adjust the transmitter power to support the efficient operation \cite{Cui2019c}.
DVS for these kinds of electronic loads can tremendously reduce the world's annual electrical power consumption \cite{david2011memory} \cite{sverdlik_2016} and significantly extend the battery life of portable devices \cite{carroll2010analysis}.

Recent literature suggests that to optimize a system's energy efficiency, DVS voltage regulators should switch among a large number of voltage levels and within a sizable voltage range \cite{peng2013instruction, hashimoto2018}. A voltage reference tracking time of the order of 10 $\mu$s or shorter is the goal for DVS for processor voltage regulator modules  \cite{burd2000dynamic}. 
This criterion is faster than most of the state\nobreakdash-of\nobreakdash-the\nobreakdash-art dc\nobreakdash-dc converters.

The LiDAR sensor is widely used as the ``eyes'' of autonomous ground and airborne vehicles \cite{Sullivan2016} because of its high accuracy in long\nobreakdash-range detection and low sensitivity to ambient interference.
Automatic power control of LiDAR is an emerging technology where the power consumption of laser transmitters are dynamically adjusted to improve the LiDAR sensors' detection accuracy and thermal management \cite{BoehmkePAUS2017, velodynehdl64E, Cui2019c}. 
Fig. \ref{fig:lidarpowersupply} is a typical LiDAR transceiver system. The reflected signal power varies according to many factors such as the reflection rate and distance of obstacles, and the environmental conditions. If the reflected signal power is near the noise equivalent power \cite{NEPDef} of the detector system, the peak of the next forward laser pulse needs to be increased to prevent overlooking obstacles. If the reflected power is too high, it may saturate the optical detector and cause the loss of measurement. This problem is traditionally solved by reducing the detector sensitivity by lowering bias \cite{Bickman2005}. However, in this method, the laser diode in the transmitter always has to operate at high power, hence it increases the junction temperature, reduces the diode lifetime, and requires more thermal management. A better alternative is to reduce the peak intensity of the laser
\cite{Osioptoelectronics2009}.\\
\begin{figure*}
    \centering
    \includegraphics[width=14cm]{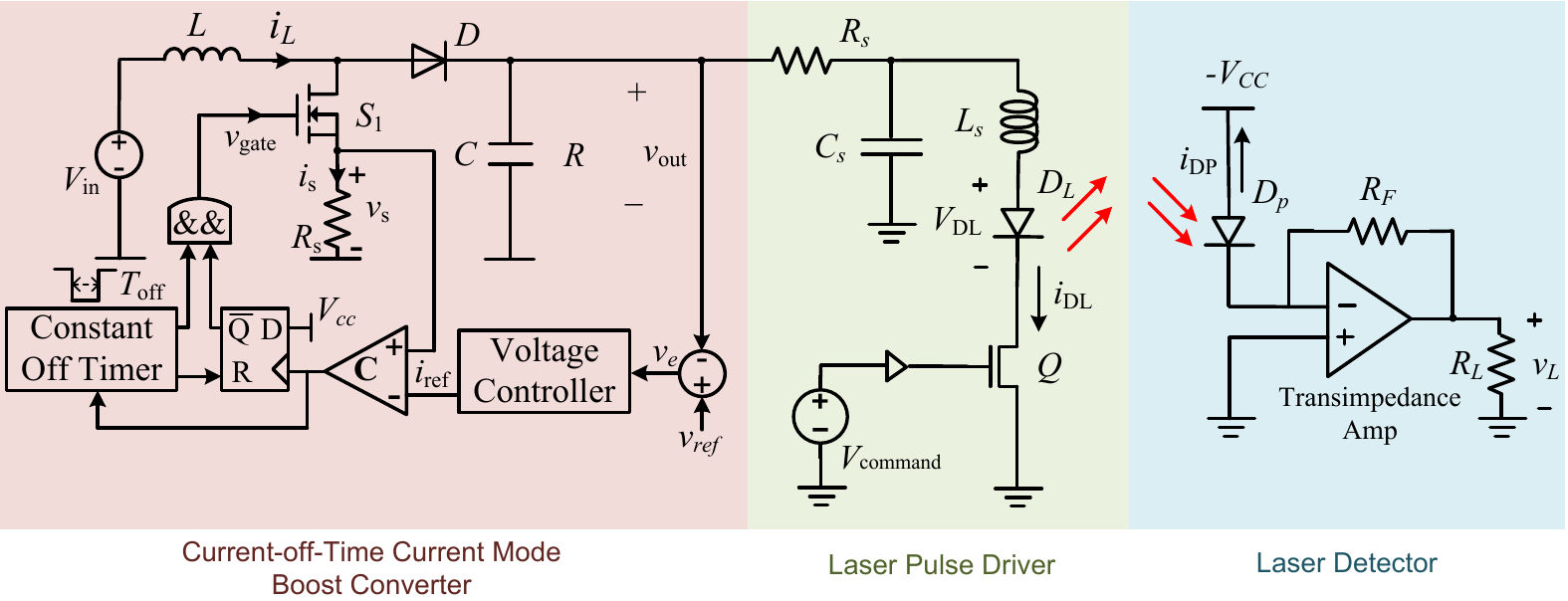}
    \caption{\label{fig:lidarpowersupply} Schematics of the constant off\nobreakdash-time, current\nobreakdash-mode boost power supply and LiDAR transceiver system.}
\end{figure*}

A boost converter is a good candidate for supplying the needed high voltage to LiDAR transmitters. A high\nobreakdash-performance boost LiDAR power supply needs a large control bandwidth to adjust the voltage level within 10$\,\mu$s \cite{velodynehdl64E} over a wide output voltage range and good load disturbance rejection to handle the instantaneous power consumption. The dynamics of power converters vary with operating point \cite{erickson2007}. 
Real\nobreakdash-time tuning can be easily realized on a digital controller to ensure consistently good dynamic performance, in contrast to a single analog compensation network.

A promising dc\nobreakdash-dc converter for DVS is 
(1) digitally controlled, (2) current\nobreakdash-mode, and (3) constant on (off)\nobreakdash-time.
First, digital control allows programmable flexibility for changes in electrical dynamics with different operating points from adaptive voltage step tracking. In comparison, a dc\nobreakdash-dc converter with an analog controller is typically built with fixed compensation; therefore, good performance cannot be always reached over a wide range of voltage and load \cite{liu2009dynamic}.
Second, current\nobreakdash-mode dc\nobreakdash-dc converters are faster and more easily compensated than voltage\nobreakdash-mode because current\nobreakdash-mode converters are first\nobreakdash-order systems \cite{erickson2007}. 
Third, constant on(off)\nobreakdash-time operation in dc\nobreakdash-dc converters does not need the additional complication of slope compensation \cite{Redl1981a} while ensuring \emph{settling} to the commanded current for all operating points \cite{cmpartone2022, cmparttwo2022}.

A common approach to the design of digital controllers for power converters is based on the framework of physical time.  The switched\nobreakdash-system is converted to a time\nobreakdash-invariant system through averaging, which is then transformed to the $s$\nobreakdash-domain to design a controller.  Bilinear transforms are used to convert this controller to the $z$\nobreakdash-domain \cite{pidaparthy2015push, priewasser2014}. Averaging strategies are more complicated for variable switching\nobreakdash-frequency converters because intervals for cycle\nobreakdash-averages are non\nobreakdash-uniform and other methods are often used \cite{Li2010}; the dynamic response of these converters are ultimately limited by  the longest switching period, which can happen during a transient.

Digital hardware complexity often increases with both dynamic response and switching frequency because of digital circuit averaging or algorithms that require high fidelity reconstruction of waveforms.  For these, a separation of time scales between the switching frequency and ADC (analog\nobreakdash-to\nobreakdash-digital converter) sampling rate is needed \cite{priewasser2014}, which often requires oversampling.

We provide a new framework to perform digital control on variable frequency power converters with faster dynamics without the computation, algorithmic, and hardware complexity in prevailing approaches. In comparison to the traditional periodic sampling and control framework, our new digital control framework shown in Fig. \ref{eventdrivenframe} includes a series of non\nobreakdash-periodic sampling and control actions, which are triggered by events instead of clocks.
\begin{figure}
\centering
\includegraphics[width=7cm]{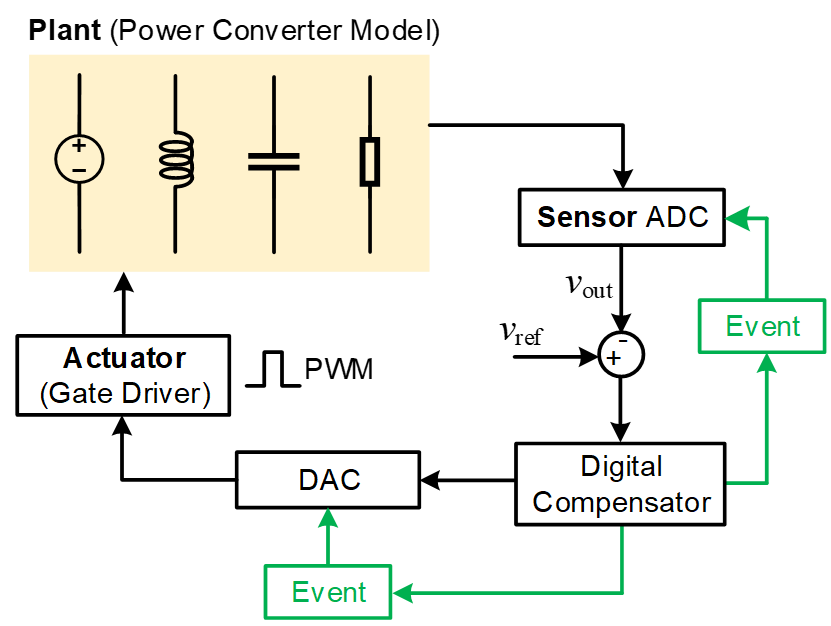}
\caption{\label{eventdrivenframe}Switching\nobreakdash-synchronized control framework for power converters.} 
\end{figure}

Unlike the traditional discrete\nobreakdash-time state space, which uses periodic sampling or interpolation to obtain a uniform correspondence to a continuous\nobreakdash-time state space, our control framework relies on a new state space called {\em switching\nobreakdash-synchronized sampled\nobreakdash-state space} ({\em 5S}) \cite{Cui2018a, xiaofanacc2019}.
In \cite{Cui2018a} and \cite{xiaofanacc2019}, the derivations use an invariant current ramp, which is valid for small output voltage steps. In this article, we advance the theory for large output voltage steps, which is a significant update.
We mathematically prove that the dynamics in {\em 5S} can fully reconstruct the physical circuit dynamics in the time domain, specifically no information is lost and there is a bijective mapping between {\em 5S} and continuous\nobreakdash-time state space.  We also mathematically prove that a controller designed in the time domain is equivalent to a controller design in {\em 5S}, under settling time and overshoot constraints; design optimization in {\em 5S} is equivalent to optimization in the time domain.

Both direct digital design and controller implementation can be performed in {\em 5S}. Because no approximate state space transformation is needed, unlike other methods, cycle\nobreakdash-by\nobreakdash-cycle digital control is precise.  Also, because the $z$\nobreakdash-transform can still be employed, power\nobreakdash-converter models can be derived as simple linear, low\nobreakdash-order systems, in contrast to more complicated plant models from the continuous\nobreakdash-time\nobreakdash-derived $s$\nobreakdash-domain.  Familiar classical control methods for direct\nobreakdash-digital design of controllers can be used.  Because sampling is only as fast as the {\em local} switching frequency, both the burden on the ADC and the digital hardware for computation is alleviated.  The naturally\nobreakdash-synchronized switch also allows sampling to be chosen to avoid switching transients.

In this paper, we present the theoretical framework and demonstrate the real\nobreakdash-life validity with a current\nobreakdash-mode buck converter with a constant on\nobreakdash-time using a switching synchronized controller.
Moreover, we apply and demonstrate the {\em 5S} controller design method in a constant\nobreakdash-off\nobreakdash-time current\nobreakdash-mode (COT\nobreakdash-CM) boost voltage regulator prototype that operates in CCM with a peak switching frequency of 3\,MHz. It is designed to nominally deliver 16\,W of power from a 12\,V vehicle battery to a 40\,V LiDAR transmitter array. The cycle\nobreakdash-by\nobreakdash-cycle digital control at this frequency and power level has not to our knowledge been reported in literature. We illustrate a COT\nobreakdash-CM boost model which matches the simulation to within 6\% error in a voltage step response test.  
The converter using the digital S\textsuperscript{2}PI controller shows a rise time of 5\,$\mu$s for a reference voltage step at several operating points ranging from 50\,\% to 100\,\% of the nominal voltage. The prototype shows a 2.5\,\% voltage deviation under a 40\,\% load step disturbance. The hardware results correspond to theory. This power supply prototype is well\nobreakdash-suited for next\nobreakdash-generation autonomous vehicle LiDAR.

This paper is organized as the following: \renewcommand\labelenumi{(\theenumi)}
\begin{enumerate*}
\item Section I introduces the paper; 
\item  Section II discusses the modeling in a switching\nobreakdash-synchronized sampled state space;
\item  Section III explains the switching\nobreakdash-synchronized sampled\nobreakdash-state space control concepts for dc\nobreakdash-dc converters;
\item  Section IV investigates the digital controller design in $\emph{5S}$;
\item Section V exhibits the hardware implementations and experimental results;
\item Section VI concludes the paper.
\end{enumerate*}
\section{Modeling in a Switching-Synchronized Sampled-State Space} \label{sec:modeling_5s}
By using a discrete\nobreakdash-event state space that is synchronized to switching actions and using a\nobreakdash-priori information about state trajectories, perfect reconstruction can be attained at sub\nobreakdash-Nyquist rates, which results in efficient sampling and control.  In this framework, each switching cycle consists of a {\em single} sampling event and a corresponding {\em single} control action. Because the switching interval is varying, sampling and control event intervals are time\nobreakdash-varying, yet always synchronized to switching events. A new sampled state\nobreakdash-space, which is extracted from these non\nobreakdash-periodically sampled states, is a departure from the traditional discrete\nobreakdash-time state space, which has a uniform correspondence to a continuous\nobreakdash-time state space. We term this new state\nobreakdash-space {\em 5S} for {\em switching\nobreakdash-synchronized sampled-state space}. We can show that the dynamics in {\em 5S} represents the true time\nobreakdash-domain circuit dynamics through reconstruction despite sampling well below the Nyquist rate. We illustrate this new digital control framework through (1) a valley\nobreakdash-current\nobreakdash-mode buck converter with constant on\nobreakdash-time and (2) peak\nobreakdash-current\nobreakdash-mode boost converter with constant off\nobreakdash-time.

\subsection{Operation of the Constant On\nobreakdash-Time Current\nobreakdash-Mode Buck Converter}
A COT\nobreakdash-CM buck converter is illustrated in Fig.\,\ref{circuitdiagram_buck}.
\begin{figure}[htbp]
	\centering
	\includegraphics[width = 7cm]{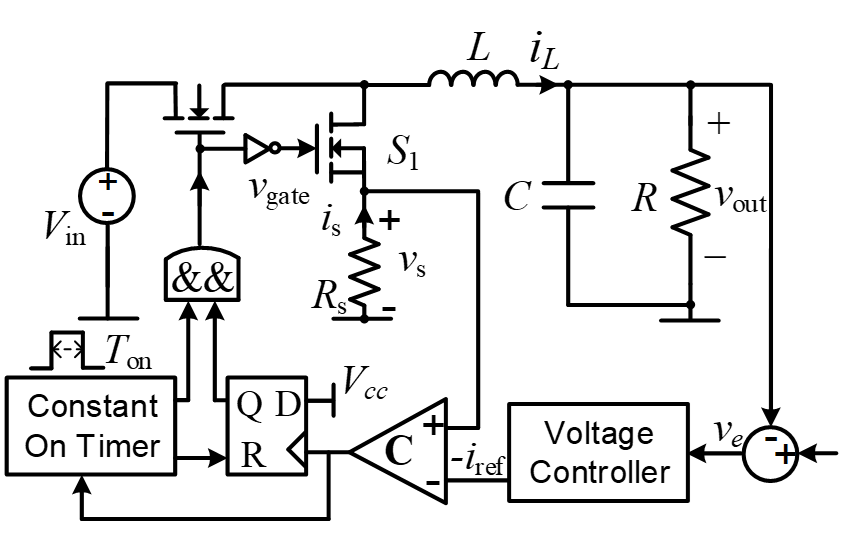}
	\caption{Circuit diagram of a digitally\nobreakdash-controlled constant\nobreakdash-on\nobreakdash-time current\nobreakdash-mode buck converter.}
	\label{circuitdiagram_buck}
\end{figure}
The annotated waveforms for inductor current and capacitor voltage are shown in Fig.\,\ref{ebuckderivation}.
\begin{figure}[htbp]
    \centering
    \includegraphics[width = 7cm]{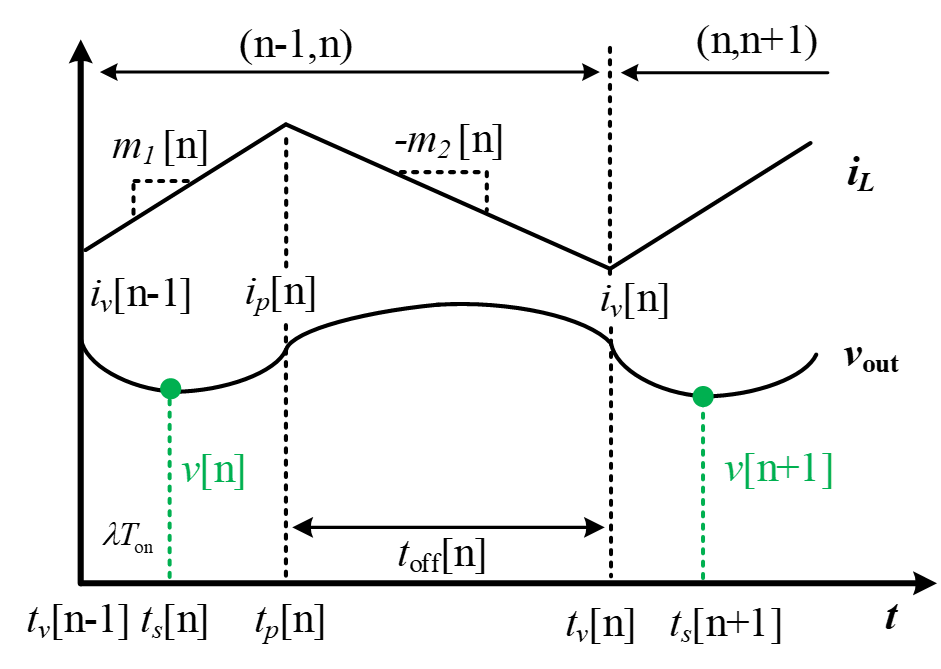}
    \caption{Inductor current and capacitor voltage waveforms of a constant\nobreakdash-on\nobreakdash-time current\nobreakdash-mode buck converter.}
    \label{ebuckderivation}
\end{figure}

The valley current, enforced by a controller at the instance $t_v[n]$, is denoted at the $n^{\text{th}}$ switching cycle as $i_v[n]$ with the peak occurring at $t_p[n]$. 
The valley\nobreakdash-current controller can be realized by a comparator and a digital\nobreakdash-to\nobreakdash-analog converter (DAC). The inductor current through $S_2$, converted to a voltage with a current\nobreakdash-sense resistor, is compared to the reference value from the DAC, which is updated every cycle.  If the valley current falls below the command, $S_2$ turns off and $S_1$ turns on. Because the valley\nobreakdash-current settles in one cycle for a constant on\nobreakdash-time controller, i.e. deadbeat, the command and actuation of current can be considered simultaneous.

The output voltage is sampled at every time instant $t_s[n]$ and is denoted for the $n^{\text{th}}$ cycle as $v[n]$. The ADC acquires and converts the output voltage during the on\nobreakdash-time of $S_1$. The sampling time instant of the ADC is parameterized by $\lambda$, specifically $ t_s[n] = t_v[n-1] +\lambda T_{\text{on}} $, $\lambda \in (0,1)$. 

The ordering of the indices for the sequences $\{i_v[n]\}$ and $\{v[n]\}$ in Fig.\,\ref{ebuckderivation} is determined by circuit topology and causality.  The reasoning is as follows. The output voltage that is sampled in the time domain immediately after $i_v[n-1]$ cannot be indexed as $v[n-1]$ because of the implication of a non-existent direct feed-forward path from current to voltage in the circuit. This same output voltage also cannot be indexed as $v[n+1]$ because causality dictates that $v[n+1]$ cannot be updated ahead of $i_v[n]$.  Only the unique ordering for the indices in Fig.\,\ref{ebuckderivation} remains.




\subsection{Usual Power Electronics Assertions for COT\nobreakdash-CM\\ Buck Converters} \label{subsec: common_asumptions}
The COT\nobreakdash-CM buck converters satisfy the following properties:
\renewcommand\labelenumi{(\theenumi)}
\begin{enumerate*} []\item the on\nobreakdash-time $T_{\text{off}}$ of $S_1$ is constant;\item off\nobreakdash-time $t_{\text{on}}[n]$ of $S_1$ is determined by valley\nobreakdash-current (the minimum inductor current every switching cycle), $0 < t_{\text{on}}[n]< +\infty$; \item output voltage $v_{\text{out}}$ has small ripple so the inductor waveform can be well\nobreakdash-approximated as a linear ramp; \item the inductor current is piecewise linear with the slopes of the rising and falling current ramps denoted as $m_1$ and $-m_2$ respectively; \item input voltage $V_{\text{in}} > 0$ is constant; \item the circuit is in continuous\nobreakdash-conduction\nobreakdash-mode (CCM) \cite{erickson2007}; \item output capacitor voltage denoted by $v(t) > 0$.
\end{enumerate*}
\subsection{Modeling of the Buck Converters in 5S} \label{sec: 5s_model_buck}

We begin by analyzing the more general problem of a linear\nobreakdash-ramp current source with intercept $i_0$ and slope $m$ charging a parallel resistor $R$ and capacitor $C$ whose initial voltage is $v_0$. The solution of capacitor voltage $v(t)$ is
\begin{equation}
v(t) = \frac{1}{C}\int_0^t (i_0 + m \tau )e^{-\frac{t-\tau}{RC}}d\tau + v_0e^{-\frac{t}{RC}}. 
\end{equation}
We can approximate this voltage, which can be representative of the output of a power converter, by a quadratic using a Taylor series approximation with an error that is small for $T_{\text{on}} \ll RC$,
\begin{equation}
v(t)  =   v_0 +   t\left(\frac{i_0}{C} -\frac{v_0}{RC}\right) + t^2\frac{m}{2C}.
\end{equation}
From this perspective, a practical current\nobreakdash-mode buck converter with constant on\nobreakdash-time  can be represented as a time\nobreakdash-piecewise\nobreakdash-linear current source charging the output of an $RC$ filter, with the rising slope $m_1$ and falling slope $m_2$ can be expressed as
\begin{eqnarray}
m_1 = \frac{V_{\text{in}}-v_{\text{out}}[n]}{L}; 
m_2 = \frac{v_{\text{out}}[n]}{L}.
\end{eqnarray}
The inductor current ramp can then be written as
\begin{align}
i_p[n] &= i_v[n-1] + m_1T_{\text{on}},\label{iv2ip_buck}\\  
i_v[n] &= i_p[n] - m_2t_{\text{off}}[n],\label{ip2iv_buck}
\end{align}
where $t_{\text{off}}[n]$ is implicitly controlled by the valley-current command.

In the time interval $(t_s[n],t_p[n])$, we treat the inductor current $i_v[n-1]+m_1t$ as a ramp that charges the output $RC$ filter. Given capacitor voltage $v[n]$  at time instance $t_s[n]$, we can express the capacitor voltage $v(t_p[n])$ by
\begin{eqnarray}
v({t_p}[n]) &=& \left(1-\frac{(1-\lambda ) T_{\text{on}} }{RC}\right)v[n]+ \frac{(1-\lambda )^2 m_1 T_{\text{on}}^2}{2 C} \nonumber \\
&\ &     + \frac{(1-\lambda ) T_{\text{on}} i_v[n-1]}{C}.
\end{eqnarray}

In the time interval $(t_p[n],t_v[n])$, given capacitor voltage $v({t_p}[n])$, we can express the capacitor voltage $v({t_v}[n])$ by
\begin{equation}
v({t_v}[n]) = \left(1-\frac{ t_{\text{off}}[n]}{RC}\right)v({t_p}[n])- \frac{ m_2 t^2_{\text{off}}[n]}{2 C}+ \frac{t_{\text{off}}[n] i_p[n]}{C}.
\end{equation}

In the time interval $(t_p[n],t_s[n+1])$, given capacitor voltage $v({t_v}[n])$, we can express the capacitor voltage  $v[n+1]$  by
\begin{equation}
v[n+1] = \left(1-\frac{\lambda T_{\text{on}}}{RC}\right)v({t_v}[n])+ \frac{\lambda ^2 m_1 T_{\text{on}}^2}{2 C}+\frac{\lambda T_{\text{on}} i_v[n]}{C}.
\end{equation}

We have the equation between $v[n+1]$, $v[n]$, $i_v[n-1]$ and $i_v[n]$. It reveals a non\nobreakdash-linear relationship between the valley current sequence and the voltage sequence. We perform a small perturbation to this curve to obtain the linearization about the operating point
\begin{equation} \label{DE_off}
\tilde{v}[n+1] = \gamma_v \tilde{v}[n] +  \gamma_i \tilde{i}_v[n] + \gamma_{im1} \tilde{i}_v[n-1].
\end{equation}
By denoting $\tilde{\tau}_1 = RC/T_{\text{on}}$, $\tilde{\tau}_2 = L/R T_{\text{on}}$, and $M_r = \left(V_{\text{in}} - V_{\text{out}} \right)/V_{\text{out}}$, the coefficients $\gamma_v$, $\gamma_i$, and $\gamma_{im1}$ could be parameterized as follows:
\begin{eqnarray} 
\gamma_v &=& 1-(1+M_r)\,\hat{\tau}^{-1}_1-\frac{1+M_r}{2}\hat{\tau}^{-1}_1\hat{\tau}^{-1}_2; \nonumber \\
\gamma_i &=& R \left(\lambda + \frac{M_r}{2}\right) \hat{\tau}^{-1}_1; \nonumber\\
\gamma_{im1} &=& R \left(1 - \lambda + \frac{M_r}{2}\right) \hat{\tau}^{-1}_1 \nonumber.
\end{eqnarray}

In {\em 5S}, the power converter can be modeled as a low\nobreakdash-order difference equation. Although {\em 5S} is obtained from a non\nobreakdash-periodic sampling process, the $z$\nobreakdash-transform can always be performed on any linearized difference equation within a region of convergence irrespective of the underlying sampling.
We perform the $z$\nobreakdash-transform on linear difference equation (\ref{DE_off}); the $z$\nobreakdash-domain expression of the plant is 
\begin{equation}
\frac{\hat v(z)}{\hat i_v(z)} = g_1\,\frac{(1-b_1z^{-1})\,z^{-1}}{1-a_1z^{-1}},
\end{equation}
where
\begin{align}
a_1 &= 1-(1+M_r)\,\hat{\tau}^{-1}_1-\frac{1+M_r}{2}\hat{\tau}^{-1}_1\hat{\tau}^{-1}_2, \nonumber \\  g_1 &= R \left(\lambda + \frac{M_r}{2}\right) \hat{\tau}^{-1}_1,\quad
b_1 = -\frac{1-\lambda+ M_r/2}{\lambda+M_r/2},\nonumber \\
M_r & = \frac{V_{\text{in}} - V_{\text{out}}}{V_{\text{out}}}, \quad \hat{\tau}_1 = \frac{RC}{T_{\text{on}}}, \quad \hat{\tau}_2 = \frac{L/R}{T_{\text{on}}}.
\end{align}

The $z$\nobreakdash-domain transfer function gives insights to controller design because it reveals the direct relationship between circuit parameters and {\em 5S} dynamics. The slow pole in the transfer function indicates the amount of time scale separation between the switch on\nobreakdash-time and the time constant of the output low\nobreakdash-pass filter. Together the fast pole and zero describe the delay between measurement and actuation; the zero  represents the delay being less than one sample period.  This zero varies with $\lambda T_{\text{on}}$, which is the delay time from the valley\nobreakdash-current time to the voltage\nobreakdash-sampling time. The pole $z=0$ corresponds to one switching cycle delay, which reflects the causality of the original physical system. In the circuit, there is no direct feed\nobreakdash-forward from the inductor current to the capacitor voltage, which means that the sampled voltage for the current cycle is determined by the current from the previous cycle.

\subsection{Operation of Constant Off\nobreakdash-Time Current\nobreakdash-Mode \\Boost Converters}
The current\nobreakdash-mode boost converter with constant off\nobreakdash-time is illustrated in Fig.\,\ref{fig:lidarpowersupply}. The term ``constant off\nobreakdash-time'' indicates that the turn\nobreakdash-off time $T_\text{off}$ of switch $S_1$ is predetermined and kept constant. The inductor current settles in one cycle after the peak current command is updated. The turn\nobreakdash-on time of $S_1$ is implicitly determined by the peak\nobreakdash-current command $i_p[n]$ at $t_p[n]$ ($n > 0$). The sampling of $v_{\text{out}}$ and control algorithms are conducted cycle\nobreakdash-by\nobreakdash-cycle during the off\nobreakdash-time because the off\nobreakdash-time is constant even during transients. $v[n]$ and $i_p[n]$ are not at the same physical time although their indices are same. 
The peak\nobreakdash-current controller uses an analog comparator to determine when the inductor current crosses a threshold that is determined by the DAC reference, which is updated every cycle. Fig.\,\ref{fig:inductorboost} shows the current and the voltage of inductor $L$ and capacitor $C$, respectively.
 \begin{figure}
    \centering
    \includegraphics[width=8cm]{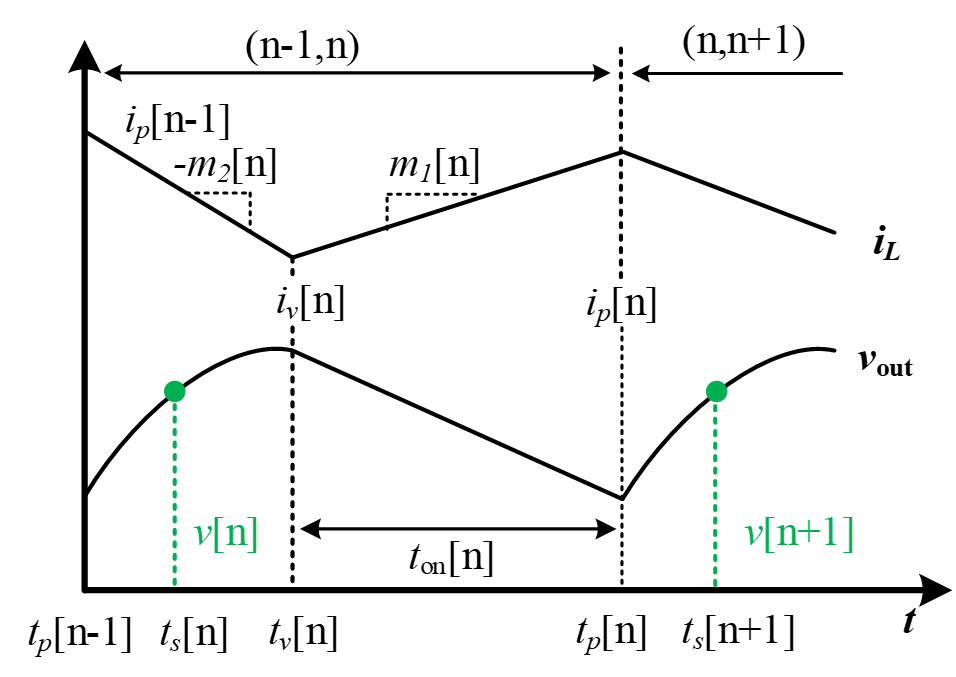}
    \caption{\label{fig:inductorboost} Inductor current and capacitor voltage waveform for COT\nobreakdash-CM boost converter.}
\end{figure}

The acquisition and the conversion of the output voltage occurs when $S_1$ is turned\nobreakdash-off. We use a fixed parameter $\lambda$ to quantify the sampling time $t_s[n]$, namely \mbox{$ t_s[n] = t_p[n-1] +\lambda T_{\text{off}} $},
\mbox{$\lambda \in (0,1)$}.
The peak inductor current during each switching cycle is represented by $i_p[n]$, while the $v[n]$ sequence and the $i_p[n]$ sequence do not correspond to the same physical times. In this discretization, all the measurements and the switching are event\nobreakdash-triggered instead of clock\nobreakdash-triggered. The peak\nobreakdash-current detection signal triggers the turn\nobreakdash-off event of $S_1$, and the off\nobreakdash-time timer drives the turn\nobreakdash-on event.

We define $\Sigma$, a class of power converters that contains all the properties of the current-mode boost converter with constant off-time as discussed above.

\subsection{Usual Power Electronics Assertions for Boost Converters}
\begin{definition} \label{circuitassumption} 
$\Sigma$ is a class of boost converters with the following \textbf{properties}: \renewcommand\labelenumi{(\theenumi)}\begin{enumerate*} []\item the off\nobreakdash-time $T_{\text{off}}$ of $S_1$ is constant;\item on\nobreakdash-time $t_{\text{on}}[n]$ of $S_1$ is determined by the peak\nobreakdash-current (the maximum commanded inductor current every switching cycle), $0 < t_{\text{on}}[n]< +\infty$; \item output voltage $v_{\text{out}}$ has small ripple so the inductor waveform can be well\nobreakdash-approximated as a linear ramp; \item input voltage $V_{\text{in}}>0$ is constant; \item the circuit is in continuous\nobreakdash-conduction\nobreakdash-mode (CCM); \item output capacitor voltage denoted by $v(t) > 0$.
\end{enumerate*}
\end{definition}

\subsection{Modeling of Boost Converters in 5S}
A stabilizing voltage controller \cite{ozbay1999introduction} uses the error between the output voltage sample and reference to adjust the command for the valley\nobreakdash-current.  The voltage controller can be realized digitally with a field\nobreakdash-programmable gate array (FPGA).  In this {\em 5S} framework, all sampling and control actions are driven by events rather than time.  The turn\nobreakdash-on event for $S_1$ is determined by the valley\nobreakdash-current detection event and the turn\nobreakdash-off event by the timeout event for the on\nobreakdash-time timer.  Because of this, the events are not in general periodically nor uniformly distributed in the physical time domain.
\begin{definition} \label{defk2}
$\mathcal{S}$ is the set of systems with class $\Sigma$ plants, and controller $K$ with noise-free output voltage measurement and peak-current actuation.
\end{definition}

\begin{definition}  \label{defk3}
Given any system $\mathcal{S}_1 \in \mathcal{S}$, under the switching-synchronized discrete state-space representation, $V$ is a class of \emph{discrete state trajectories}, which includes a sequence of vectors $\left\{\mathbf{u}[n]\right\}$ with $\mathbf{u}[n-1] = \begin{bmatrix}
v[n]&  i_p[n-1]
\end{bmatrix}^{\text{T}}$. $v[n]$ is the output voltage that is sampled at $t_s[n] = t_p[n-1]+ \lambda T_{\text{off}}$, where $t_p[n-1]$ is the $n^{\text{th}}$ turn-off time of $S_1$, as in Fig.~\ref{fig:inductorboost}. 
$i_p[n]$ is the peak current after $n$ switching cycle, as in Fig.~\ref{fig:lidarpowersupply} and Fig.~\ref{fig:inductorboost}.
\end{definition}

The \emph{discrete state trajectory} is defined in \cite{Cui2018a}. We want to derive the difference equation(s) that governs the class $V$ discrete state trajectories.

By performing a similar procedure to \ref{sec: 5s_model_buck}, the $z$\nobreakdash-domain model of the boost converter is
\begin{align} {\label{eqn:modifiedplantmodel}}
\frac{\tilde v(z)}{\tilde i_p(z)} = g_1\,\frac{(1 - b_1 z^{-1})\,z^{-1}}{1 - a_1 z^{-1}}\,,
\end{align}
which is parameterized by
\begin{align}
    a_1 =\,&
    1-2(\hat{\tau}_1^{-1}+\hat{\tau}_3^{-1})-\frac{\lambda^2+(1-\lambda)^2}{2}\hat{\tau}_1^{-1} \hat{\tau}_2^{-1},  \\
    d_1 =\,& ( \lambda \hat{\tau}_1^{-1} + \frac{\lambda^2}{2}\hat{\tau}_1^{-1} \hat{\tau}_2^{-1} - 1)(\hat{\tau}_1^{-1} + \hat{\tau}_3^{-1}) \nonumber \\
    & -\left(1+(1-\lambda)\hat{\tau}_1^{-1} -2(\hat{\tau}_1^{-1}+\hat{\tau}_3^{-1}) - \frac{\lambda^2}{2}\hat{\tau}_1^{-1}\hat{\tau}_2^{-1} \right) \nonumber \\
    &\times (1-\lambda)\hat{\tau}_1^{-1} \hat{\tau}_2^{-1},  \\
    d_2 =\,& \left(\lambda \hat{\tau}_1^{-1} +\frac{\lambda^2}{2}\hat{\tau}_1^{-1} \hat{\tau}_2^{-1}-1\right)(\hat{\tau}_1^{-1} + \hat{\tau}_3^{-1}) + \lambda \hat{\tau}_1^{-1} \hat{\tau}_2^{-1}, \\
    g_1 =\,& \left(\lambda \hat{\tau}_1^{-1} - \left(1 - \lambda \hat{\tau}_1^{-1} - \frac{\lambda^2}{2} \hat{\tau}_1^{-1} \hat{\tau}_2^{-1}\right)\frac{\hat{\tau}_1^{-1}+\hat{\tau}_3^{-1}}{\hat{\tau}_2^{-1}}\right)R, \\
    b_1 =\,& \frac{d_1}{d_2}.
\end{align}

The definitions of $\hat{\tau}_1$, $\hat{\tau}_2$, and $\hat{\tau}_3$ are
\begin{align}
    \hat{\tau}_1 =\, \frac{RC}{T_{\text{off}}}, \quad \hat{\tau}_2 = \frac{L/R}{T_{\text{off}}}, \quad \hat{\tau}_3 = \frac{RC}{T_{\text{on}}}, 
\end{align}
where $T_{\text{off}}$ is the constant off-time, $T_{\text{on}}$ is the on\nobreakdash-time, and $T=T_{\text{off}}+T_{\text{on}}$. The derivation details can be found in \cite{Cui2019c}. 

 Note that the transfer function has one zero and two poles while the traditional $s$\nobreakdash-domain current\nobreakdash-mode boost converter model in \cite{Stankovic2001} shows only one zero and one pole. All poles are inside the unit disk, so the open loop system in {\em 5S} is stable. The fast pole is from the single\nobreakdash-cycle delay between measurement and actuation, which is the causality requirement of the discretized system. The slow pole is from the output $RC$\nobreakdash-filter of the boost converter, which corresponds with the traditional $s$\nobreakdash-domain model.
 The zero in the {\em 5S} model is analogous to the right\nobreakdash-half\nobreakdash-plane (RHP) zero of the averaged\nobreakdash-state\nobreakdash-space boost converter model. One difference between the {\em 5S} model and the averaged state\nobreakdash-space model is that this zero is not only determined by the circuit parameters and operating point, but also influenced by the sampling delay.
 
 The theoretical voltage step response matches the simulation, which is shown in Fig. \ref{fig:compinductor} and Fig. \ref{fig:compcapacitor} in Section \ref{section:sim and exp}. 
 We define the single\nobreakdash-step voltage error $e[n]$ between the model and simulation as
\begin{align} \label{eqn:singlesteperrordef}
    e[n] \triangleq \frac{ \left(V_{\text{sim}}[n] - V_{\text{model}}[n] \right)}{V_{\text{stepsize}}}\times 100\%.
\end{align}
Fig.\,\ref{fig:errorwpertubarion} compares the worst\nobreakdash-case voltage error between the model and simulation under different output\nobreakdash-voltage step sizes at the initial output voltage 40\,V. We define the worst\nobreakdash-case error $e_w$ between the model and simulation as
\begin{align} \label{eqn:errordef}
    e_w \triangleq \underset{n}{\mathrm{max}}\, \big|e[n]\big|.
\end{align}
\begin{figure}
    \centering
    \includegraphics[width=7cm]{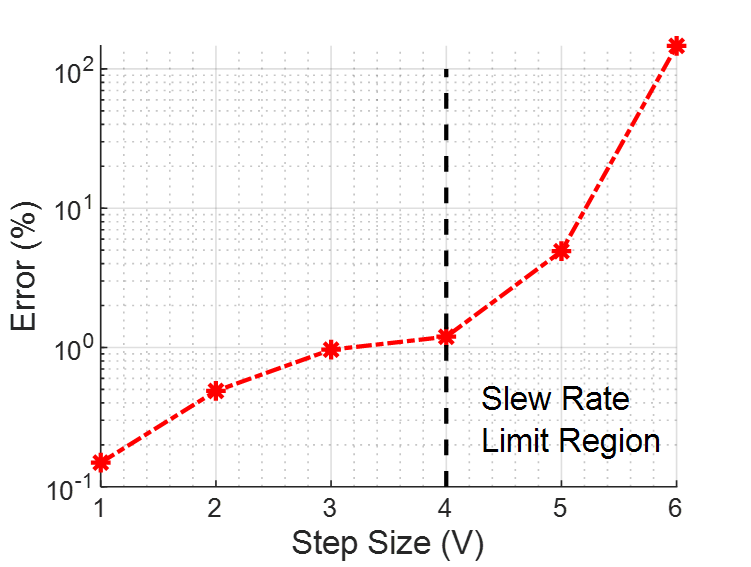}
    \caption{\label{fig:errorwpertubarion} Worst\nobreakdash-case voltage error between the model and simulation under different voltage\nobreakdash-step sizes at the initial output voltage 40\,V.}
\end{figure}

We observe that the worst\nobreakdash-case error is small (below 1\,\%) if the voltage\nobreakdash-step size is within 4\,V. The error rises quickly when the step size is above 4\,V because of large\nobreakdash-signal effects; for example, the inductor\nobreakdash-current slew rate limit starts to dominate the dynamics. For this COT\nobreakdash-CM boost converter, the first significant large\nobreakdash-signal behavior which designers face is that the minimum on\nobreakdash-time is reached; hence, the actuator is saturated and the peak inductor current no longer matches the current command.
A gain\nobreakdash-scheduled controller is needed to perform the large\nobreakdash-signal voltage reference steps, which will be discussed in details in Section\,\ref{section:sim and exp}.
\section{Switching-Synchronized Sampled-State Space Control Concepts for DC-DC Converters} \label{controllerconcept_5s}
In this section, we show that the \emph {5S} model of dc-dc converters agrees with the switched state-space model \cite{Stankovic2001} concerning the stability and performance by a COT-CM boost converter example. However, the controller synthesis problem in {\em 5S} model is much easier than that for a switched state-space model because the former is a pure discrete system while the latter is a hybrid system \cite{Albea2015}.
\subsection{Stability}
Stability is the minimum requirement for a closed\nobreakdash-loop controlled system.
Trajectories in \emph{5S} contain less information than the physical trajectories.
Because of the discretization, the information about the system dynamics between the adjacent sampled data is lost.
Because of the non\nobreakdash-uniform sampled data, the information about physical time is lost as well.
However, by utilizing the restrictions on the trajectories in Section \ref{subsec: common_asumptions}, the stability relationship is preserved.
We mathematically proved that stability in \emph{5S} enforces stability in physical time.

Stability means the capacitor voltage and inductor current waveforms converge to a periodic trajectory in the phase plane. This trajectory must contain no subharmonics of the switching frequency.

\begin{definition} \label {classWtraj}
Given any system $\mathcal{S}_1 \in \mathcal{S}$, $W$ is a class of \emph{continuous-time state trajectories} that  a vector-valued function evolving through time $\mathbf{x}(t)= \begin{bmatrix}
v_{\text{out}}(t) & i_L(t) \end{bmatrix}^{\text{T}}$. $v_{\text{out}}(t)$ is the output voltage and $i_L(t)$ is the inductor current as in Fig. \ref{fig:lidarpowersupply}. Without loss of generality and because of causality, we denote the starting time of $\mathbf{x}(t)$ as the time $t_0 \triangleq t_p[n=0]$ of the first inductor current peak.
\end{definition}

Given a system $\mathcal{S}_1 \in \mathcal{S}$ with all the parameters and initial states determined, class $V$ and class $W$ trajectories are denoted as $\{\mathbf{u}[n]\}$ and $\mathbf{x}(t)$ respectively, there must exist more than one mapping that transforms $\{\mathbf{u}[n]\}$ to $\mathbf{x}(t)$. To examine the relationship between $\{\mathbf{u}[n]\}$ and $\mathbf{x}(t)$, we only need to construct one of these mappings. 
A boost converter operating in CCM can be modeled as a switched-affine time-invariant system. In what follows, $ s \in \{0,1\}$ as a super/subscript indicates the state of the switch: $s = 0$ represents $S_1$ is off and $s = 1$ represents $S_1$ is on, as in Fig.~\ref{fig:lidarpowersupply}.

The state transition matrix $\mathbf{\Phi}_s$ and the offset vector $\mathbf{\Gamma}_s$ are defined as follows. During the time interval $\left[t_p[n],t_p[n]+T_{\text{off}}\right]$ ($n \ge 0$), 
\begin{align}   \label{swoffstatetransfer}
\mathbf{x}(t) &= \begin{bmatrix}
x_1(t_p[n]) + \alpha_1(t-t_p[n]) - \frac{m_2[n]}{2C} (t-t_p[n])^2 \nonumber \\
x_2(t_p[n]) - m_2[n] (t-t_p[n])
\end{bmatrix} \\ 
&= \mathbf{\Phi}_0(t,t_p[n])\;\mathbf{x}({t_p}[n])+\mathbf{\Gamma}_0(t,t_p[n]).
\end{align}
where $\alpha_1 = x_2(t_p[n])/C-x_1(t_p[n])/RC$,
\begin{align*}
\scalebox{0.7}{$
\mathbf{\Phi}_0(t_2,t_1) \triangleq \begin{bmatrix}
1 - \frac{t_2-t_1}{RC} 
- \frac{(t_2-t_1)^2}{2LC}
& \frac{t_2-t_1}{C} \\
-\frac{(t_2 - t_1)}{2L} & 1
\end{bmatrix}, \mathbf{\Gamma}_0(t_2,t_1) \triangleq  \begin{bmatrix}
\frac{V_{\text{in}}}{2LC}(t_2-t_1)^2\\
\frac{V_{\text{in}}}{2L}(t_2-t_1)
\end{bmatrix}$}.
\end{align*}


During the time interval $\left[t_p[n]+T_{\text{off}},t_p[n+1]\right]$,
\begingroup
\small
\begin{align}   \label{swonstatetransfer}
\mathbf{x}(t) &= \begin{bmatrix}
x_1(t_p[n]+T_{\text{off}}) (1 -  (t-t_p[n]-T_{\text{off}})/RC) \nonumber \\
x_2(t_p[n]+T_{\text{off}}) + m_1 (t-t_p[n]-T_{\text{off}})
\end{bmatrix} \\ 
&=\mathbf{\Phi}_1(t,t_p[n]+T_{\text{off}})\;\mathbf{x}(t_p[n]+T_{\text{off}}) + \mathbf{\Gamma}_1(t,t_p[n]+T_{\text{off}}),
\end{align}
\endgroup
where
\begin{align*}
\scalebox{0.95}{$
\mathbf{\Phi}_1(t_2,t_1) \triangleq \begin{bmatrix}
1-\frac{t_2-t_1}{RC} & 0\\
0 & 1
\end{bmatrix}, \mathbf{\Gamma}_1(t_2,t_1) \triangleq \begin{bmatrix}
0\\
\frac{V_{\text{in}}}{L}(t_2-t_1)
\end{bmatrix}$}.
\end{align*}
We observe that $\mathbf{\Phi}_s(t_2,t_1) = \mathbf{\Phi}_s(t_2-t_1,0)$ and $\mathbf{\Gamma}_s(t_2,t_1) = \mathbf{\Gamma}_s(t_2-t_1,0)$. We can write $\mathbf{\Phi}_s$ and $\mathbf{\Gamma}_s$ by \begin{equation}
\mathbf{\Phi}_s = \begin{bmatrix}
\phi_{vv}^s & \phi_{vi}^s \\
\phi_{iv}^s & \phi_{ii}^s
\end{bmatrix} ,\quad \mathbf{\Gamma}_s = \begin{bmatrix}
\gamma_{gv}^s\\
\gamma_{gi}^s 
\end{bmatrix}.    \nonumber 
\end{equation}

From (\ref{swoffstatetransfer}), the state at the sampling time $t_s$ in the $(n+1)^{\text{th}}$ switching cycle is
\begin{equation} \label{v2s}
\mathbf{x}({t_s}[n+1]) = \mathbf{\Phi}_0(\lambda T_{\text{off}},0)\;\mathbf{x}({t_p}[n]) + \mathbf{\Gamma}_0(\lambda T_{\text{off}},0).
\end{equation}
In defining $v(t_s[n])$ as the output voltage at the sampling time instant in the $n^{\text{th}}$ interval and $v[n]$ as the sampled output voltage at $t_s[n]$, $v(t_s[n+1]) = v[n+1]$. Similarly, $i(t_{p}[n]) = i_p[n]$ from Definition~\ref{defk3}.

From (\ref{v2s}), the state at the $(n+1)^{\text{th}}$ turn-off time instant is
\begin{align} \label{reconstruct:vts}
\mathbf{x}({t_p}[n]) = \mathbf{B} \;\mathbf{u}[n] + \mathbf{C},
\end{align}
where
\begingroup
\small
\begin{equation*}
\mathbf{B}  \triangleq
\begin{bmatrix}
\frac{1}{\phi_{vv}^0(\lambda T_{\text{off}},0)}& - \frac{\phi_{vi}^0(\lambda T_{\text{off}},0)}{\phi_{vv}^0(\lambda T_{\text{off}},0)} \\
0 & 1
\end{bmatrix}, \quad \mathbf{C} \triangleq \begin{bmatrix}
-\frac{\gamma_{gv}^0(\lambda T_{\text{off}},0)}{\phi_{vv}^0(\lambda T_{\text{off}},0)} \\ 0
\end{bmatrix}.
\end{equation*}
\endgroup
Because $\phi_{vv}(t,0) \neq 0$ for all $t>0$ from property 3 of class $\Sigma$ power converters,  $\mathbf{B}$ is an invertible matrix.\\
Any continuous-time state during the off-time of the $(n+1)^{\text{th}}$ switching cycle is determined by (\ref{swoffstatetransfer}) and (\ref{reconstruct:vts}). 

The state at the $(n+1)^{\text{th}}$ turn-on time instant of $S_1$ can be reconstructed by
\begin{equation}\label{xintp}
\mathbf{x}({t_p}[n]+T_{\text{off}}) = \mathbf{\Phi}_0(T_{\text{off}},0)\;\mathbf{x}({t_p}[n]) + \mathbf{\Gamma}_0(T_{\text{off}},0).
\end{equation} 
Based on property 5 and property 6 of class $\Sigma$ power converters, $t_{\text{on}}[n+1]$ can be expressed as
\begin{align} \label{offtimerecon}
t_{\text{on}}[n+1]= & \; ( i_p[n+1]-i_p[n]+m_2T_{\text{off}})/m_1 \nonumber \\
= & \; \mathbf{E}\mathbf{u}[n+1] + \mathbf{F}\mathbf{u}[n] + T_{\text{off}}, \end{align}
where 
\begin{align}
\mathbf{E} = \begin{bmatrix}
0 & L/V_{\text{in}}
\end{bmatrix}, \quad
\mathbf{F} = \begin{bmatrix}
T_{\text{off}}/V_{\text{in}} & -L/V_{\text{in}}.
\end{bmatrix}
\end{align}
Any continuous-time state during the on-time of the $(n+1)^{\text{th}}$ switching cycle can be determined from (\ref{swonstatetransfer}), (\ref{xintp}) and (\ref{offtimerecon}).
A mapping $\mathbf{A}$ from $\left\{\mathbf{u}[n]\right\}$ to $\mathbf{x}(t)$ can be explicitly constructed as
\begingroup
\small
\begin{align}  \label{mappingA}
 \mathbf{A}: \quad (a)&\  t_{\text{on}}[n+1]=  \; \mathbf{E}\mathbf{u}[n+1] + \mathbf{F}\mathbf{u}[n] + T_{\text{off}}. \nonumber \\
(b)&\ \mathbf{x}(\tau+\sum_{j=1}^{n} (T_{\text{off}}+t_{\text{on}}[j]) + t_0) \nonumber \\
 &= \left\{
\begin{array}{l}
   \mathbf{\Phi}_0(\tau,0)\;\mathbf{B} \;\mathbf{u}[n]  + \mathbf{\Phi}_0(\tau,0)\;\mathbf{C} + \mathbf{\Gamma}_0 (\tau,0) ,\\
  \quad  \quad  \quad  \quad \quad  \quad  \quad \quad \quad \quad \quad 
   \quad  \quad   
  0 < \tau \le T_{\text{off}}, \\
  \\
     \mathbf{\Phi}_1(\tau,T_{\text{off}}) \mathbf{\Phi}_0(T_{\text{off}},0)\;(\mathbf{B} \;\mathbf{u}[n] + \mathbf{C}) \\ 
      + \mathbf{\Phi}_1(\tau,T_{\text{off}}) \mathbf{\Gamma}_0 (T_{\text{off}},0) + \mathbf{\Gamma}_1 (\tau, T_{\text{off}}), \\
    \quad  \quad  \quad  \quad \quad \quad \quad \;
    T_{\text{off}} < \tau \le T_{\text{off}} + t_{\text{on}}[n+1].
\end{array} 
\right.
\end{align}
\endgroup
To simplify the (\ref{mappingA})-(b), we define functions $\mathbf{D}_{n+1}$ from $\mathbf{u}[n] \times \tau$ to $x(t)$ as
\begin{align}  \label{mappingD}
\mathbf{D}_{n+1} (\mathbf{u}[n],\tau) = \mathbf{x}(\tau + \sum_{j=1}^{n} (T_{\text{off}}&+t_{\text{on}}[j])+t_0),\nonumber \\
&0 <\tau \le t_{\text{on}} [n+1]. 
\end{align}

We can now show that mapping $\mathbf{A}$ preserves  stability. 
This result is consequential because in {\em 5S}, classical discrete-domain design methods can be applied to controller design.
\begin{definition}
From \cite{Boyd2004}, the distance from a point $\mathbf{x}
$ to a nonempty set $\boldsymbol{\gamma}$ is 
\begin{equation*}
  \textbf{dist}(\mathbf{x},\boldsymbol{\gamma}) \triangleq \text{inf}\{\;\norm{\mathbf{x} - \mathbf{y}}_2\;|\; \mathbf{x}\in \mathbb{R}^n, \mathbf{y} \in \boldsymbol{\gamma} \subseteq \mathbb{R}^n\}.
\end{equation*}
\end{definition}
\begin{definition} \label{defeq} 
For a system $\mathcal{S}_1 \in \mathcal{S}$, a continuous-time state trajectory $\mathbf{x}_e(t)$ is a \emph{least harmonics equilibrium} if there exists a period $T>0$ such that $\mathbf{x}_e(t+T) = \mathbf{x}_e(t)$ for all $t \ge 0$ and there is only one current peak during the period $T$.\footnote{In this paper, unless specially mentioned, the equilibrium means least harmonics equilibrium.}
\end{definition}
\begin{definition} \label{defstability} 
$\mathbf{x}_e(t)$ is an equilibrium of system $\mathcal{S}_1 \in \mathcal{S}$. Assume $\mathbf{x}_e(t)$ is perturbed at $t_b$ and the perturbed trajectory is $\mathbf{x}(t)$. $t_1 \ge t_b$ and $t_2 \ge t_b$ occur at the current peaks of $\mathbf{x}(t)$ and $\mathbf{x}_e(t)$ respectively. $\mathbf{x}_e(t)$ is a \emph{synchronously asymptotically stable least harmonics equilibrium of $S$} if there exists a $\delta>0$ such that if 
$\norm[2]{\mathbf{x}(t_1)-\mathbf{x}_e(t_2)}_2<\delta$, 
then $\displaystyle\lim_{t\rightarrow +\infty}  \textbf{dist}(\mathbf{x},\mathbf{x}_e(t)) = 0$.
\end{definition}

\begin{theorem}
\label{stabilityrelationship}
For a system $\mathcal{S}_1 \in \mathcal{S}$, if a discrete state in a switching-synchronized sampled-state space representation $\mathbf{u}_e$ is an asymptotically stable equilibrium in the sense of Lyapunov (ISL) and if we construct a sequence $\{\mathbf{u}_e[n]\}$ with $\mathbf{u}_e[n]=\mathbf{u}_e = [V_e, I_e]$, $\forall n\ge 0$, then the corresponding  continuous-time state trajectory $ \mathbf{x}_e(t) \triangleq \mathbf{A}(\{\mathbf{u}_e[n]\})$ is a stable equilibrium in the sense of Definition \ref{defstability}.
\end{theorem}
\begin{proof}

We denote the induced 2-norm \cite{horn1990matrix} for matrix $\mathbf{Z}$ by $\norm[1]{\mathbf{Z}}_M \triangleq \displaystyle\sup_{\mathbf{x} \neq \mathbf{0}}\frac{\|\mathbf{Z}\mathbf{x}\|_{2}}{\|\mathbf{x}\|_2}$.

i) We show  $\mathbf{x}_e(t)$ is an equilibrium. From (\ref{offtimerecon}), $ t_{\text{on}}^{\mathbf{e}}  \triangleq \left((v_e- V_{\text{in}})/V_{\text{in}}\right)T_{\text{off}} = t_{\text{on}}[n+1]$ for all $ n \ge 0$. Let $T = t_{\text{on}}^{\mathbf{e}}+T_{\text{off}}$. From (\ref{mappingA}) and (\ref{mappingD}), given any $ n \ge 0$,
\begin{align}
\mathbf{x}_e(\tau+nT+t_0) =& \mathbf{D}_{n+1}(\mathbf{u}[n],\tau)  \nonumber \\
=& \mathbf{D}_1(\mathbf{u}[0]=\mathbf{u}_e,\tau) = \mathbf{x}_e(\tau+t_0),
\end{align}
where $ 0<\tau\le T$. Therefore,  $\mathbf{x}(t+T) =  \mathbf{x}(t), \forall t \ge t_0$, where $t_0$ is from Definition \ref{classWtraj}. Because there is only one peak current during each $T$ period, $\mathbf{x}_e(t)$ is an equilibrium.

ii) We show that $\mathbf{x}_e(t)$ is an asymptotically stable equilibrium. Because $\mathbf{u}_e$ is an asymptotically stable equilibrium ISL, $\exists \; \delta_1>0$ such that if $\norm[1]{\mathbf{u}[0] - \mathbf{u}_e}_2<\delta_1$, then $\displaystyle\lim_{n\rightarrow +\infty} \mathbf{u}[n] =  \mathbf{u}_e$, where $\{\mathbf{u}[n]\}$ is the corresponding class $V$ trajectory after the perturbation. The corresponding class $W$ trajectory is $\mathbf{x}(t) = \mathbf{A}(\{\mathbf{u}[n]\})$.
For any perturbation on $\mathbf{x}_e(t)$ at $t=t_b$, $t_1=t_p[0] \ge t_b$ at a current peak of $x(t)$ and $t_2 = t_p^{\mathbf{e}}[0] \ge t^b$ at a current peak of $x_e(t)$ are both at discrete state space points.
\begin{equation}
\norm[1]{\mathbf{x}(t_1)-\mathbf{x}_{e}(t_2)}_2 =  \norm[1]{\mathbf{B}(\mathbf{u}[0] - \mathbf{u}_e)}_2 \le \norm[1]{\mathbf{B}}_M \norm[1] {\delta_1}_2.
\end{equation}
We show that if $x(t_1)$ lies in the $\delta_2$ neighborhood of $x_e(t)$ where $\delta_2 = \norm[1]{\mathbf{B}}_M \norm[1] {\delta_1}_2$, then given any $\epsilon_1 > 0$, there exists a $t_c > \text{max}\{t_1, t_2\}$ such that $\mathbf{dist}(\mathbf{x},\mathbf{x}_e(t)) \le \epsilon_1$ for all $t \ge t_c$.


Because every entry of $\mathbf{\Phi}_1(t,0)$, $\mathbf{\Phi}^{'}_1(t,0)$ and  $\mathbf{\Gamma}^{'}_1(t,0)$ are continuous functions of $t$, from Theorem 2.5.4 in \cite{horn1990matrix}, $\norm[1]{\mathbf{\Phi}_1(t,0)}_M$, $\norm[1]{\mathbf{\Phi}^{'}_1(t,0)}_M$, and  $\norm[1]{\mathbf{\Gamma}^{'}_1(t,0)}_2$ are continuous functions of $t$.

From the \emph{Weierstrass Theorem} in \cite{horn1990matrix}, in the closed interval $\left[0,2 t^{e}_{\text{on}}\right]$, there exists $M_1, L_1$ and $L_2$ such that
\begin{equation}  \label{boundml}
\norm[1]{\mathbf{\Phi}_1(t,0)}_M \leq M_1, \; \norm[1]{\mathbf{\Phi}_1^{'}(t,0)}_M \leq L_1, \; \norm[1]{\mathbf{\Gamma }_1^{'}(t,0)}_2 \leq L_2.
\end{equation}

Given any $\epsilon_1 > 0$, $\exists N_1$, $ \forall n \ge N_1$ such that
\begin{align}
\norm[1]{\mathbf{u}[n]-\mathbf{u}_e}_2 \le \text{min}\biggl\{ & \frac{\epsilon_1}{2M_1\norm[1]{\mathbf{\Phi}_0(T_{\text{off}},0)\mathbf{B}}}_M, \nonumber \\
\frac{V_{\text{in}}}{2L + T_{\text{off}}}t^{e}_{\text{on}}, \, & \frac{m_1 \epsilon_1}{2L_1(a_0+a_1+a_2)+L_2)} \biggl\},
\end{align}
where $a_0 = \norm[1]{\mathbf{\Phi}_0(T_{\text{off}},0)\mathbf{B}}_M\;(\ \norm[1]{\mathbf{u}_e}_2+\delta_1) $, $ a_1 = \norm[1]{\mathbf{\Phi}_0(T_{\text{off}},0)\mathbf{C}}_2 $ and $ a_2 = \norm[1]{\mathbf{\Gamma}_0(T_{\text{off}},0)}_2 $.

 Given any $t > \sum_{j=1}^{N_1} (T_{\text{off}}+t_{\text{on}}[j]) $, there exists $N_2 \ge N_1$ such that $\sum_{j=1}^{N_2} (T_{\text{off}}+t_{\text{on}}[j]) < t \le \sum_{j=1}^{N_2+1}(T_{\text{off}}+t_{\text{on}}[j])$. Let $\tau = t-\sum_{j=1}^{N_2} (T_{\text{off}}+t_{\text{on}}[j])$ and then $ 0 < \tau \leq T_{\text{off}} + t_{\text{on}}[N_2+1]$.

$t_{\text{on}}[N_2+1]$ is bounded in the closed interval $\left[0,2 t^{e}_{\text{on}}\right]$ because
\begin{align} \label{boundu}
  & \biggr | t_{\text{on}}[N_2+1] - t^{e}_{\text{on}}\biggr| \le \nonumber \\ 
  & \left| \frac{i_p [N_2 + 1] -i_p[N_2]}{m_1} \right| + \left| \left(\frac{v_{\text{out}}[N_2] - V_e}{m_1 L}\right) T_{\text{off}}\right| \le \nonumber \\
  & \frac{L}{V_{\text{in}}} \norm[1]{\mathbf{u}[N_2+1]-\mathbf{u}_e}_2 + \frac{L}{V_{\text{in}}} \norm[1]{\mathbf{u}[N_2]-\mathbf{u}_e}_2 \nonumber \\ 
  & + \frac{T_{\text{off}}}{V_{\text{in}}} \norm[1]{\mathbf{u}[N_2]-\mathbf{u}_e}_2 \le t^{e}_{\text{on}}.
\end{align}

If $\tau \leq T$, then
\begin{align}
&\textbf{dist}(\mathbf{x},\mathbf{x}_e(t)) \leq \norm[1] {\mathbf{\Phi}_1(\tau, T_{\text{off}}) \mathbf{\Phi}_0(T_{\text{off}},0)\mathbf{B}(\mathbf{u}[N_2]-\mathbf{u}_e)}_2 \leq \nonumber \\
&M_1  \norm[1]{\mathbf{\Phi}_0(T_{\text{off}},0) \mathbf{B}}_M \norm{\mathbf{u}[N_2]-\mathbf{u}_e}_2 < \epsilon_1.
\end{align}

If $\tau > T$, then
\begin{align} \label{ineq10}
&\textbf{dist}(\mathbf{x},\mathbf{x}_e(t)) \leq 
\|
\mathbf{\Phi}_1(T,T_{\text{off}})\mathbf{\Phi}_0(T_{\text{off}},0)\mathbf{B}(\mathbf{u}[N_2]-\mathbf{u}_e) \nonumber \\
&+\left(\mathbf{\Phi}_1(\tau,T_{\text{off}})-\mathbf{\Phi}_1(T,T_{\text{off}})\right)\mathbf{\Phi}_0(T_{\text{off}},0) (\mathbf{B}\mathbf{u}[N_2]+\mathbf{C}) \nonumber \\
&+\left(\mathbf{\Phi}_1(\tau,T_{\text{off}})-\mathbf{\Phi}_1(T,T_{\text{off}})\right)\mathbf{\Gamma}_0(T_{\text{off}},0) \nonumber \\ 
&+ \mathbf{\Gamma}_1(\tau,T_{\text{off}})-\mathbf{\Gamma}_1(T,T_{\text{off}})
\|_2.
\end{align}
From the \emph{Mean Value Theorem} in \cite{dieudonne2013foundations}, 
\begin{align} \label{mvt}
\norm[1]{\mathbf{\Phi}_1(\tau,T_{\text{off}})-\mathbf{\Phi}_1(T,T_{\text{off}})}_M \le L_1 |\tau - T|, \nonumber\\
\norm[1]{\mathbf{\Gamma}_1(\tau,T_{\text{off}})-\mathbf{\Gamma}_1(T,T_{\text{off}})}_2 \le L_2 |\tau - T|.
\end{align}
From (\ref{ineq10}) and (\ref{mvt}),
\begin{align}
\textbf{dist}(\mathbf{x},&\mathbf{x}_e(t)) \leq M_1  \norm[1]{\mathbf{\Phi}_0(T_{\text{off}},0) \mathbf{B}}_M \norm[1]{\mathbf{u}[N_2]-\mathbf{u}_e}_2 \nonumber\\
&+ \; (\tau - T)  L_1   \norm[1]     {\mathbf{\Phi}_0(T_{\text{off}},0)\mathbf{B}}_M\;(\ \norm[1]{\mathbf{u}_e}_2+\delta_1) \nonumber \\
&+ \; (\tau - T) L_1 \norm[1]{\mathbf{\Phi}_0(T_{\text{off}},0) \mathbf{C}}_2 \nonumber\\
&+ \; (\tau - T) L_1  \norm[1]{\mathbf{\Gamma}_0(T_{\text{off}},0)}_2 + (\tau - T)L_2 
\leq \epsilon_1.
\end{align}

Hence, we proved that for any $\epsilon_1 > 0$, there exists a $t_c = \sum_{j=1}^{N_1} (T_{\text{off}}+t_{\text{on}}[j])$ such that $\mathbf{dist}(\mathbf{x},\mathbf{x}_e(t)) \le \epsilon_1$ for all $t \ge t_c$. This implies $\displaystyle\lim_{t\rightarrow +\infty}  \textbf{dist}(\mathbf{x},\mathbf{x}_e(t)) = \boldsymbol{0}$.

In all, we proved that for $\mathbf{x}_e(t)$, $ \exists \ \delta_2$, such that if the perturbation on $\mathbf{x}_e(t)$ at $t_b$ results in $\norm[1]{\mathbf{x}(t_1)-\mathbf{x}(t_2)}_2<\delta_2$, then $\displaystyle\lim_{t\rightarrow +\infty}  \textbf{dist}(\mathbf{x},\mathbf{x}_e(t)) = \boldsymbol{0}$.

From Definition \ref{defstability}, the continuous-time state trajectory $\mathbf{x}_e(t)$ is an asymptotically stable least harmonics equilibrium of $\mathcal{S}$.
\end{proof}

The stability criterion in {\em 5S} is identical to the classical discrete\nobreakdash-time system stability theory. From root locus or Nyquist \cite{Kuo1970} stability criteria, we can provide a sufficient and necessary condition for a compensated discrete linear system: all closed\nobreakdash-loop poles are inside the open unit disk. 
\subsection{Transient Response}
The transient performance of a controlled system is usually characterized by the response to a step input \cite{kuo1987automatic}. 
For constant on(off)\nobreakdash-time dc\nobreakdash-dc converters, we analyze the output voltage step response. We focus on settling time and overshoot because they are among the two most important performance criteria for power converter designers.
In this section, we show that the control performance in \emph{5S} maps to the physical time.
\begin{figure}
    \centering
    \includegraphics[width = 7cm]{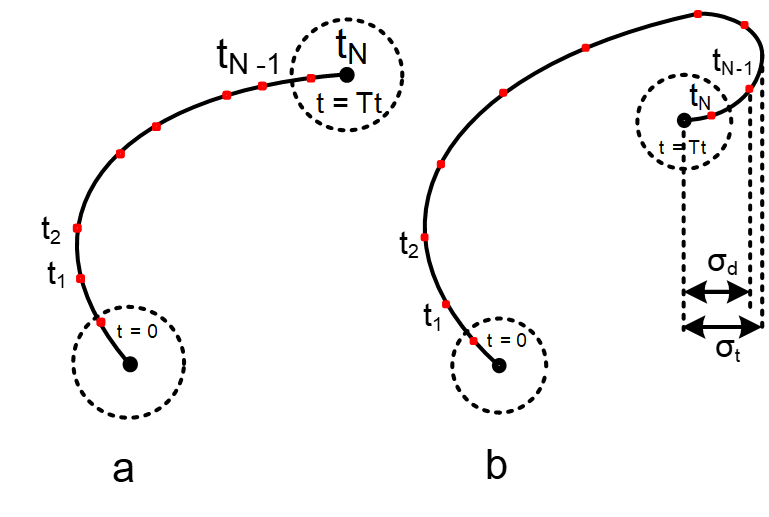}
    \caption{\label{fig:settling_overshoot} (a) Optimizing the transient cycles is equivalent to optimizing the transient time; (b) overshoot in continuous time is bounded from the top by the overshoot in 5S.}
\end{figure}
Given a closed\nobreakdash-loop constant on(off)\nobreakdash-time dc\nobreakdash-dc converter system with a reference output voltage step command $ V_{e1} \rightarrow  V_{e2}$, the \emph{settling time} $T_t$ is the transition time from one steady state to another and the \emph{settling cycles} $N_t$ is the number of turn\nobreakdash-on instants from one steady state to another.
The \emph{overshoot} for the sampled output\nobreakdash-voltage, peak inductor current, and output voltage is given by 
\begin{align}
    \label{eqn: v_is_d} \sigma_{d}^v & \triangleq \frac{\underset{n}{\mathrm{max}}\; v[n] -  V_{e2}}{V_{e2} - V_{e1}},\\ 
    \label{eqn: i_is_d}  \sigma_{d}^i & \triangleq \frac{\underset{n}{\mathrm{max}}\;i[n] -  I_{e2}}{I_{e2}-I_{e1}}, \\
    \label{eqn: v_is_t} \sigma_{t}^v & \triangleq \frac{\underset{t}{\mathrm{max}}\; v(t) - V_{e2}}{V_{e2} - V_{e1}}.
\end{align}

Because of the non\nobreakdash-uniform sampling, the settling steps and settling time are not proportional as illustrated in Fig.\,\ref{fig:settling_overshoot}(a). 
We mathematically proved that the optimization on the \emph{settling steps} in the \emph{5S} representation is equivalent to the optimization on the worst\nobreakdash-case \emph{settling time} in the physical state space. 

\begin{theorem}
\label{theoremsettlingtime}
Given a system $\mathcal{S}_1 \in \mathcal{S}$ and
$\mathbf{u}_e$ is an asymptotically stable equilibrium ISL.
Given a system $\mathcal{S}_1 \in \mathcal{S}$ in the reference output voltage step command $ V_{e1} \rightarrow  V_{e2}$,
the settling time $T_t$ is bounded by the following functions of the settling cycles $N_t$ 
\begin{align}
    T_{t} \le \rho N_{t} + \gamma,
\end{align}
where 
\begin{align}
    \rho(\sigma_d^v) &= T_1 + \frac{V_{e2}-V_{e1}}{V_{\text{in}}}T_{\text{off}}\sigma_d^v, \\
    \gamma &= \frac{L}{V_{\text{in}}}\left(I_{e2} - I_{e1} \right).
\end{align}
\end{theorem}
\begin{proof}
Assume $\mathbf{u}_e=\begin{bmatrix}
V_e & I_e
\end{bmatrix}$ and the initial state is $\mathbf{u}[0]=[V_0,I_0]$.
\begin{align} \label{eq:1}
	T_t  = \sum_{k=0}^{N_t-1} (T_{\text{off}}+t_{\text{on}}[k+1]).
\end{align}
Based on property 5 and property 6 of a class $\Sigma$ power converters
\begin{align} \label{eq:2}
    \frac{V_{\text{in}}}{L} t_{\text{on}} [n] - \frac{v[n] - V_{\text{in}}}{L} T_{\text{off}} = i_p[n] - i_p[n-1].
\end{align}
Summing up both sides of (\ref{eq:2}) for $n = 1$ to $n = N_t$ results in
\begin{align} \label{eqn: sum_delta_i}
    \frac{V_{\text{in}}}{L} \left(T_t - N_t T_{\text{off}} \right) - \frac{ \sum_{k=1}^{N_t} v[k] - N_t V_{\text{in}}}{L} T_{\text{off}} = I_{e2} - I_{e1}.
\end{align}
From (\ref{eqn: v_is_d})
\begin{align} \label{eqn: sumv_os}
    \sum_{k=1}^{N_t} v[k] \le N_t V_{e1} + \sigma_d^v N_t \left(V_{e2} - V_{e1}\right)
\end{align}
Substituting (\ref{eqn: sumv_os}) into (\ref{eqn: sum_delta_i}) yields
\begin{align}
    T_t  \le & \frac{L}{V_{\text{in}}} \left( I_{e2} - I_{e1} \right) + T_{\text{off}}N_s + \frac{V_{\text{e1}} - V_{\text{in}}}{V_{\text{in}}} T_{\text{off}}N_s \nonumber \\
    & + \sigma_d^v \frac{V_{e2}-V_{e1}}{V_{\text{in}}}T_{\text{off}} N_s \nonumber \\ 
    = & \left(T_1 + \frac{V_{e2}-V_{e1}}{V_{\text{in}}}T_{\text{off}}\sigma_d^v \right) N_s + \frac{L}{V_{\text{in}}} \left( I_{e2} - I_{e1} \right)
\end{align}
\end{proof}

Because of the discretization, the overshoot in \emph{5S} is smaller than that in continuous time as illustrated in Fig.\,\ref{fig:settling_overshoot}(b).
We mathematically proved that the \emph{overshoot} in the physical time is bounded from the top by a linear function of the \emph{overshoot} in \emph{5S}.

\begin{theorem}\label{OvershootTheorem}
Given a system $\mathcal{S}_1 \in \mathcal{S}$ in a reference output voltage step command  $ V_{e1} \rightarrow  V_{e2}$, the overshoot in continuous time $\sigma_{t}^v$ is bounded by the following functions of the overshoot in $\emph{5S}$ $\sigma_{d}^v$
\begin{align}
  \sigma_{t}^v \le (1-(1-\lambda)\alpha) \sigma_{d}^v + (1-\lambda)\alpha \sigma_{d}^i,
\end{align}
where $\alpha = T_{\text{off}}/RC$.
\end{theorem}

\begin{proof}
From (\ref{eqn: v_is_d}) and (\ref{eqn: v_is_t}), $\sigma_{d}^v \leq \sigma_{t}^v$.
We want to bound $\sigma_{t}^v$ from above. Let $t_m={\mathrm{argmax}}\;v(t)$, $N={\mathrm{argmax}}\;v[n]$ and $K={\mathrm{argmax}} \; i_p[n]$. Assume $t_m$ 
is in the $M^{\text{th}}$ switching cycle. Let $\tau = t - t_p[M-1]$ and $\alpha_c = (1-\lambda)T_{\text{off}}/RC$.

In time interval $ 0 < \tau \leq T_{\text{off}}$, the capacitor voltage $v(\tau)$ increases with $\tau$. The charging current is the difference between charge injected by the current source and charge drained by the load. During the interval $ T_{\text{off}} < \tau \leq T_{\text{off}} + t_{\text{on}} [M] $, the capacitor is discharging to the load; therefore, the capacitor voltage $ v(\tau)$ decreases with $\tau$. Let $i_c(t)$ be capacitor current, then 
the maximum $v(t)$ is bounded by
\begin{align} \label{voltbound}
v(t_m) = & \;v[M] + \int_{\lambda T_{\text{off}}}^{T_{\text{off}}} \frac{i_c(\tau)}{C}\; d\tau  \nonumber \\
\le & v[M] + \int_{\lambda T_{\text{off}}}^{T_{\text{off}}} \frac{1}{C} \left( i_L(\tau) - \frac{v[M]}{R} \right)\; d\tau \nonumber \\
\le & \left( 1 - \alpha_c \right)v[N] + \frac{1}{C} \int_{\lambda T_{\text{off}}}^{T_{\text{off}}} i_L(\tau) d\tau \nonumber\\
\le & \;\left( 1 - \alpha_c \right)v[N] + (i_p[K]-I_{e1})\alpha_c R  + I_{e1}\alpha_c R.
\end{align}
From (\ref{eqn: i_is_d})
\begin{align} \label{ipk}
i_p[K]-I_{e1} =  (\sigma_d^i+1)(I_{e2}-I_{e1})
= (\sigma_d^i+1)\frac{V_{e2}-V_{e1}}{R},
\end{align}
\vspace{-15pt}
\begin{align}  \label{ipk2}
    V_{e1} =  I_{e1} R.
\end{align}
Substitute (\ref{ipk}) and (\ref{ipk2}) into (\ref{voltbound}) yields
\begin{align}  \label{voltagebound}
    & v(t_m) - V_{e1} \nonumber \\
    \le &  \left(1 - \alpha_c \right) \left( v[N] - V_{e1} \right)
    + (\sigma_d^i+1)(V_{e2}-V_{e1})\alpha_c. 
\end{align}
Substitute (\ref{eqn: v_is_d}) and (\ref{eqn: v_is_t}) into (\ref{voltagebound})
\begin{align}
\sigma_t^v \le & \; (1-\alpha_c)\sigma_d^v + \alpha_c \sigma_d^i.
\end{align}
\end{proof}

Therefore, the digital controller design in the time domain can be conveniently converted to a controller design in {\em 5S}. In other words, we easily transform a complicated hybrid\footnote{Combined continuous\nobreakdash-time states and discrete\nobreakdash-time states \cite{tabuada2009verification}.} optimization problem to a pure discrete\nobreakdash-domain optimization problem, which is much easier to design.
\section{Digital Controller Design in {\em 5S}} \label{sec:Digital_Controller_Des}
\subsection{Digital Controller Design for a COT\nobreakdash-CM Buck Converter}
We have shown that performance\nobreakdash-optimized
digital controller design in the time domain can be equivalently converted to a performance\nobreakdash-optimized controller design in {\em 5S}. Controller design problems in {\em 5S} can be solved by classical control methods. Root locus or Nyquist plots can be used for direct\nobreakdash-digital design of the controller because the plant is modeled as a $z$\nobreakdash-domain transfer function \cite{Kuo1970}. In a current\nobreakdash-mode buck converter with constant on\nobreakdash-time, the $z$\nobreakdash-domain transfer function is a second\nobreakdash-order system with one fast pole and one slow pole. For the controller, we add a pole at $z = 1$ for zero steady\nobreakdash-state error with a zero to comprise a proportional\nobreakdash-integral (PI) controller, which is all that is needed.

\begin{figure}
\centering
\includegraphics[width=8cm]{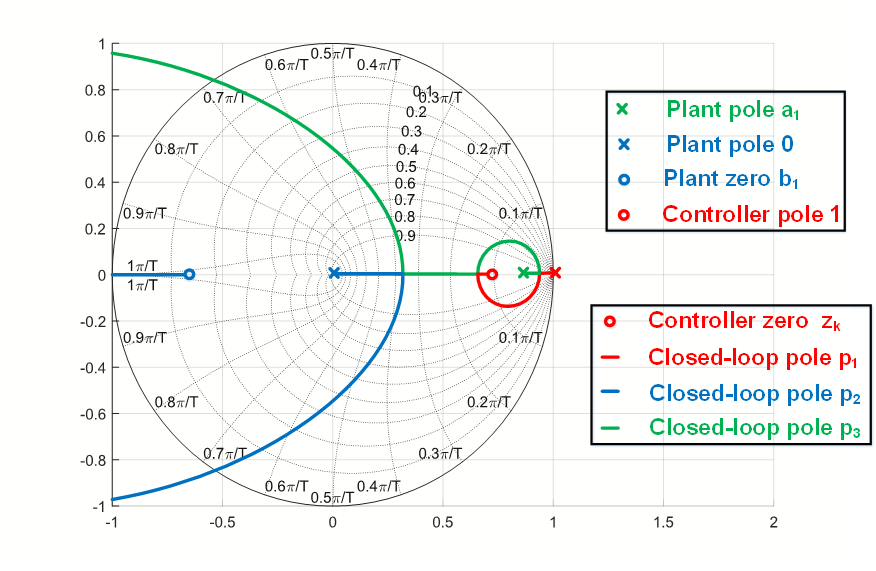}
\caption{\label{fig:rootlocus_buck} Closed\nobreakdash-loop root\nobreakdash-locus of PI compensated constant\nobreakdash-on\nobreakdash-time current\nobreakdash-mode controlled buck converter.} 
\end{figure}
We show a controller design example here for a current\nobreakdash-mode buck converter with constant on\nobreakdash-time with the hardware prototype specifications chosen with the goal of minimizing settling time with no overshoot. We design in {\em 5S} with the equivalent goal of minimizing the number of cycles for settling under a zero overshoot constraint.  The plant model is $0.0023(1+1.4390z^{-1})z^{-1}/(1-0.9739z^{-1})$ with a controller pole at $a=1$ for zero steady\nobreakdash-state error, zero at $b=0.9750$, and gain of $62$. The root locus is illustrated in Fig.\,\ref{fig:rootlocus_buck}, where it is apparent that the closed\nobreakdash-loop dynamics are dominated by a complex\nobreakdash-conjugate pole pair.
We verify the discrete\nobreakdash-time step later with the hardware.

\subsection{Digital Controller Design for a COT\nobreakdash-CM Boost Converter}
The open\nobreakdash-loop converter plant for a COT\nobreakdash-CM boost converter does not have a transient response that is fast enough to perform DVS for LiDAR because the output $RC$\nobreakdash-filter results in a slow pole. This shortcoming motivates us to develop a systematic compensation method. 

Any traditional discrete\nobreakdash-time controller can be transformed to {\em 5S} using an event\nobreakdash-driven sampler and actuator, which are updated synchronously with switching actions.
We illustrate a design example of a boost converter performing a voltage step tracking task. We use the root\nobreakdash-locus method as shown in Fig. \ref{fig:rootlocus_boost}. In our compensator, we first place a pole \mbox{$p_k = 1$} (integrator) to force the steady\nobreakdash-state voltage\nobreakdash-tracking error to be zero. Also, we need a zero $z_k$ as well as a gain $k$ to accelerate the transient response. The resulting {\em switching\nobreakdash-synchronized} proportional\nobreakdash-integral (S\textsuperscript{2}PI) compensator can be expressed as
\begin{align}
    K(z) = k \frac{1- z_k z^{-1}}{1-p_kz^{-1}}.
\end{align}

We compensate the slow open\nobreakdash-loop pole $a_1$ by placing a zero $z_k$ in its neighborhood.
The perceived optimal way is to achieve pole\nobreakdash-zero cancellation; however, this cancellation is never perfect in practice because of the uncertainty of $a_1$.
We choose $z_k$ to the left of the poles $z = a_1$ and $z = 1$.
From root locus rules, $z_k$ is the destination of the root locus leaving $a_1$.
The resulting closed\nobreakdash-loop pole $p_1$ stays on the real axis and in the neighborhood of $z_k$.
At high gain, the settling is determined by the zero $z_k$.
\begin{figure}
    \centering
    \includegraphics[width=8cm]{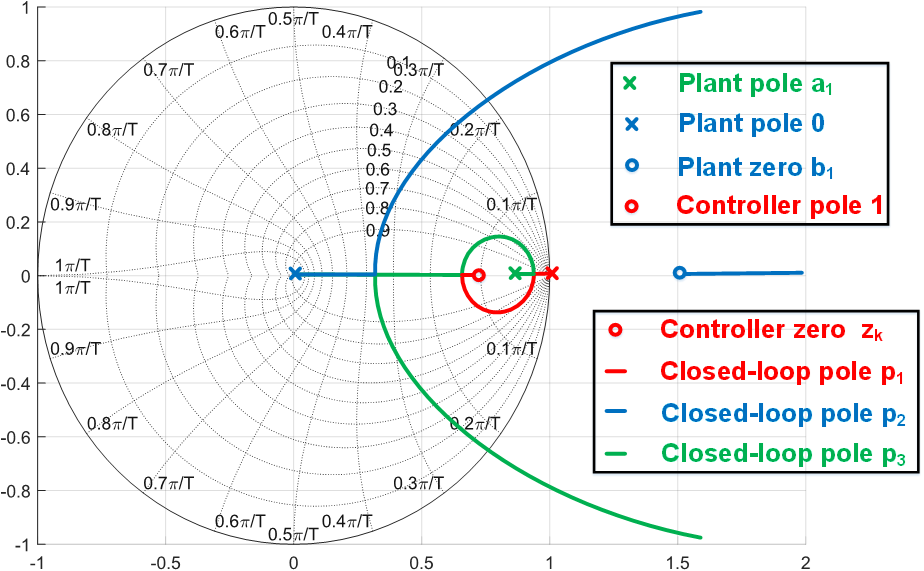}
    \caption{\label{fig:rootlocus_boost} 
     Root\nobreakdash-locus method for digital controller design in {\em 5S} of a COT\nobreakdash-CM boost converter.}
\end{figure}

\section{Hardware and Experimental Results} \label{sec:hardwaredesignguide}
\subsection{Buck Converter for VRM Applications}
We designed and built a current\nobreakdash-mode buck converter with constant on\nobreakdash-time that is controlled by a {\em 5S} digital controller, which is a prototype for dynamic voltage scaling and is shown in Fig. \ref{fig:prototype}.
The 1.8\,V output was selected based on the power requirements for a typical microprocessor \cite{TI2017}. The power level was chosen to be 20\,W based on the TDP (Thermal Design Power) value \cite{Intel2017}.
The switching frequency was chosen to be nominally 1.13\,MHz to demonstrate cycle\nobreakdash-by\nobreakdash-cycle digital control which can be challenging.  
Other circuit parameters are reported in Table \ref{t1}.
\begin{figure*}[htbp]
\centering
\includegraphics[width=12cm]{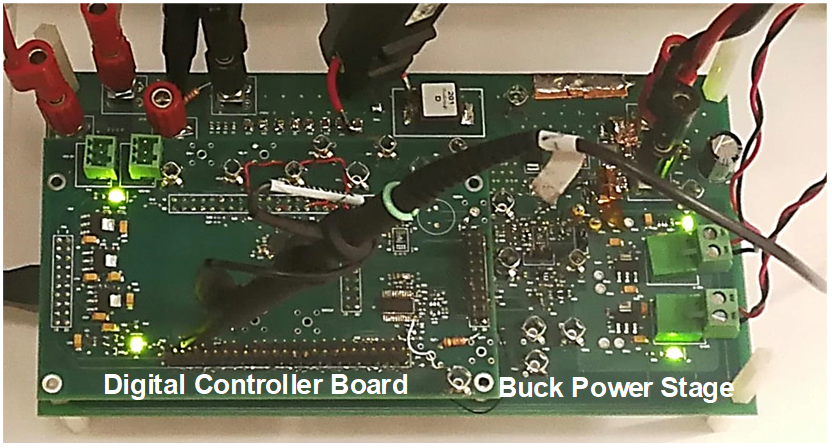}
\caption{\label{fig:prototype} Current-mode constant on-time buck converter prototype and digital control hardware.}
\end{figure*}
\begin{table}[htbp]
    \caption{Buck converter design parameters}
    \label{t1}
    \centering
    \begin{tabular}{cccccccccccccc}
        \toprule
            \textbf{Param.}&\textbf{Values}&\textbf{Param.}&\textbf{Values}\\
            \midrule
            $\boldsymbol{V_{\textbf{in}}}$
            & 8\,V & $\boldsymbol{L}$ & 200\,nH\\
            \midrule
            $\boldsymbol{V_{\textbf{out}}}$
            & 1.8\,V & Power & 20\,W \\
            \midrule
            $\boldsymbol{T_{\textbf{on}}}$ 
            & 100\,ns & $\boldsymbol{R_s}$ & 10\,m$\Omega$ \\
            \midrule
            $\boldsymbol{C}$ & 200 $\mu$F & $\boldsymbol{f_{sw}}$
            & 1.13\,MHz \\
            \midrule
            FPGA & Spartan-6 & ADC
            & LTC2378-16 \\
            \midrule
            MOSFET & IRF6620 & DAC & MAX5184 \\
            \midrule
            Diode & B520 & Diff. Op-Amp & LT1994 \\
        \bottomrule
    \end{tabular}
\end{table}


The prototype performs well for fast transient response in both small\nobreakdash-signal reference step and large\nobreakdash-signal reference step. In Fig.\,\ref{SS_exp}, a small\nobreakdash-signal step of 50\,mV shows a rise time as fast as 5\,$\mu$s with no overshoot.
\begin{figure}[htbp]
\centering
\includegraphics[width=\columnwidth]{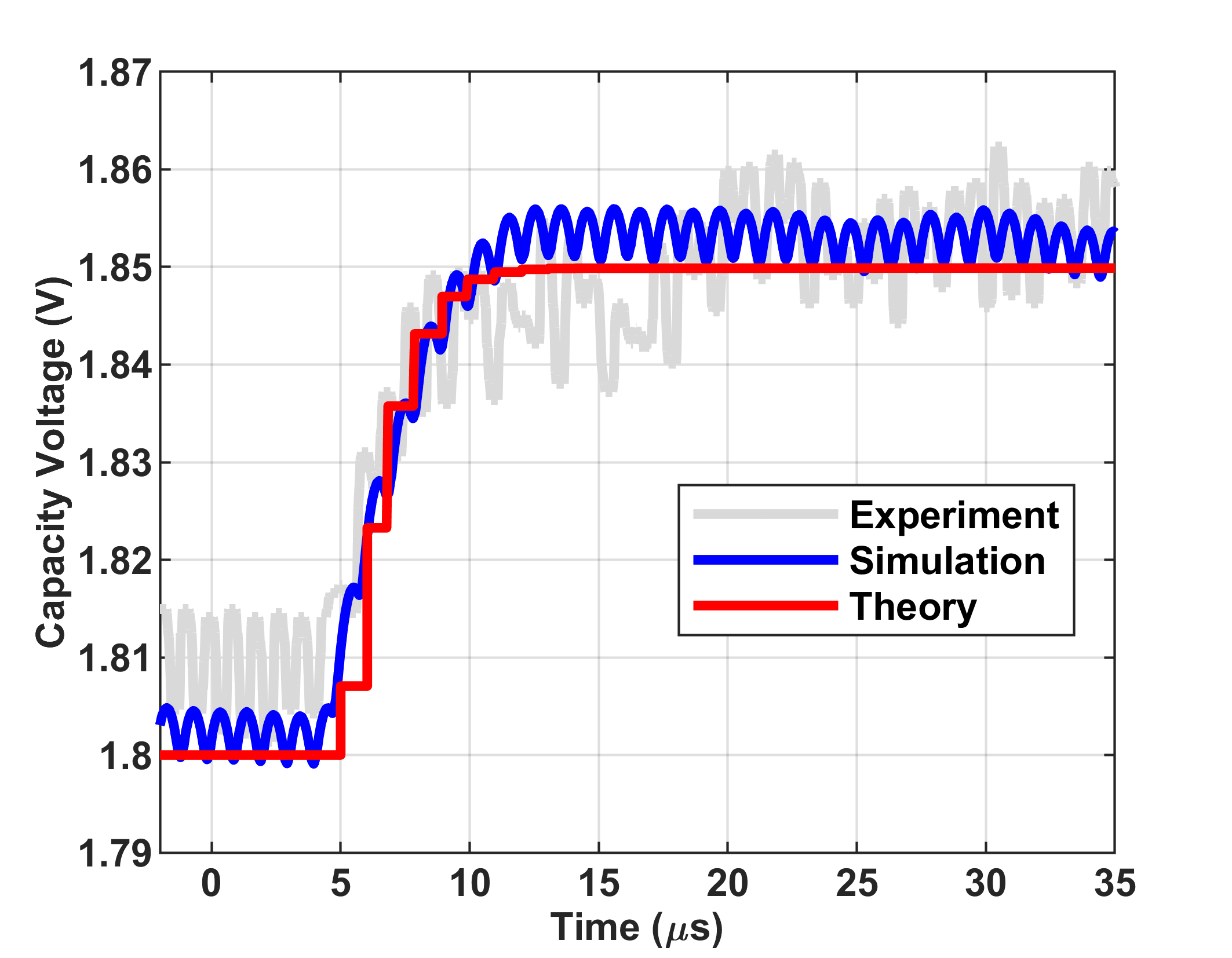}
\caption{\label{SS_exp} Small\nobreakdash-signal output capacitor voltage  response of a CM\nobreakdash-COT buck converter in reference voltage step.}
\end{figure}
\begin{figure}[htbp]
\centering
\includegraphics[width=\columnwidth]{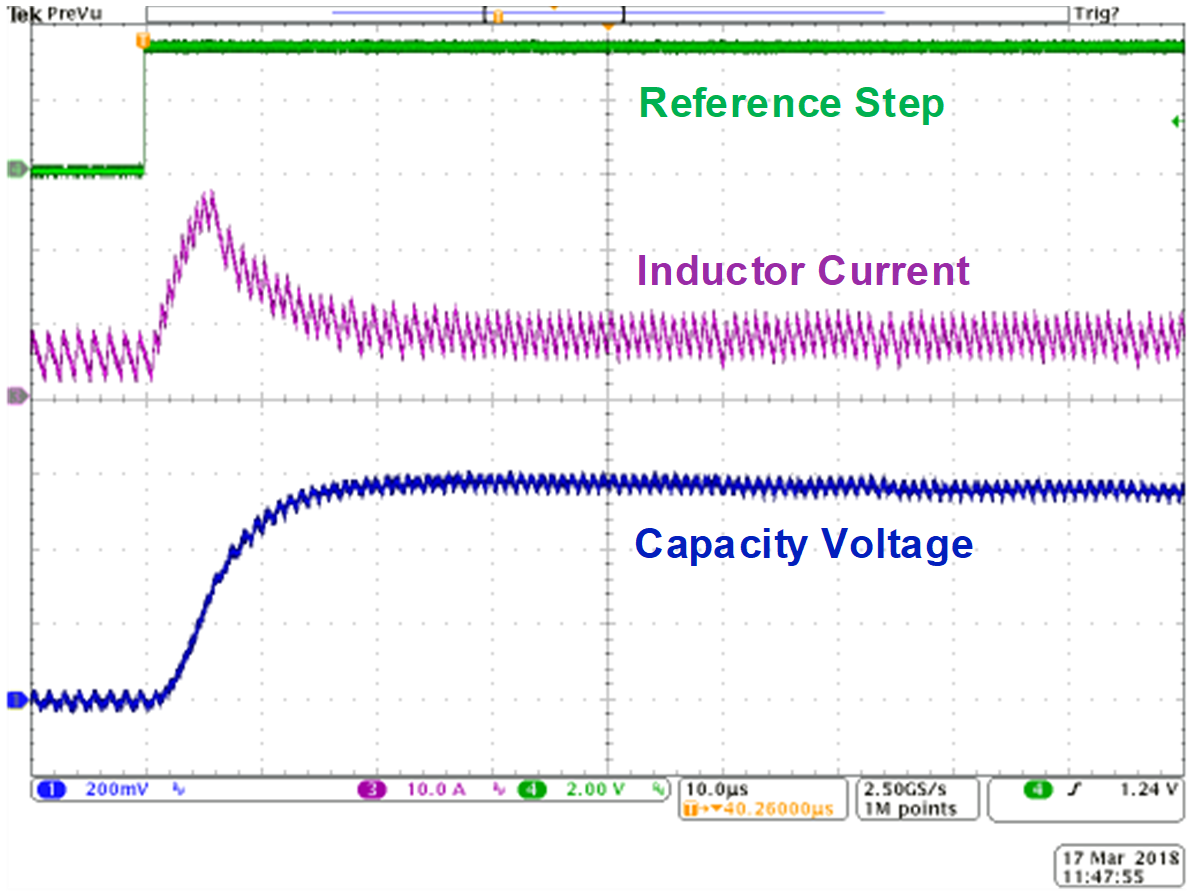}
\caption{\label{LS_exp} Large\nobreakdash-signal output voltage/inductor current response of a CM\nobreakdash-COT buck converter in reference voltage step.}
\end{figure}
Fig. \ref{SS_exp} shows good agreement between theory, simulation and experiment. The possible reason for the small discrepancy between theory and experiment might be attributed to our model assumption of a lossless circuit with real losses causing deviations in the duty ratio and switching frequency from the ideal. 

A large\nobreakdash-signal 0.5 V reference step to a 1.8 V set point was demonstrated to be stable despite inductor slew rate limiting. 
The large\nobreakdash-signal response in Fig.\,\ref{LS_exp} shows a rise time as fast as 8 $\mu$s with less than 3$\%$ overshoot.

\subsection{Boost Converters for LiDAR Applications} \label{section:sim and exp}
We designed and built a current\nobreakdash-mode boost converter with constant off\nobreakdash-time that is controlled by a {\em 5S} digital controller. Our COT\nobreakdash-CM boost regulator includes an analog peak\nobreakdash-current\nobreakdash-control circuit and digital voltage\nobreakdash-control loop as shown in Fig.\,\ref{fig:lidarpowersupply}.
The prototype shown in Fig.\,\ref{fig:prototype_compel2019} was constructed with the parameters in Table\,\ref{tab:boosttable} and controlled by an Artix\nobreakdash-7 FPGA from Xilinx with a 400\,MHz system clock. The control algorithm follows the flowchart in Fig. \ref{fig:flowchart}. One control cycle consists of approximately 80 FPGA clock cycles which determines the 200\,ns constant off\nobreakdash-time. The 12\,V input is a common voltage level in a vehicle. The output voltage was selected to be 40\,V based on the laser driver solution in \cite{Diodes2018}. The power level was set to be 16\,W based on a commercial product \cite{velodynehdl64E}. The peak switching frequency is 3\,MHz because LiDAR transmitters need to be more compact and portable \cite{velodynehdl64E} with high switching frequency largely shrinking the size and weight.

\begin{figure}[htbp]
    \centering
    \includegraphics[width=8cm]{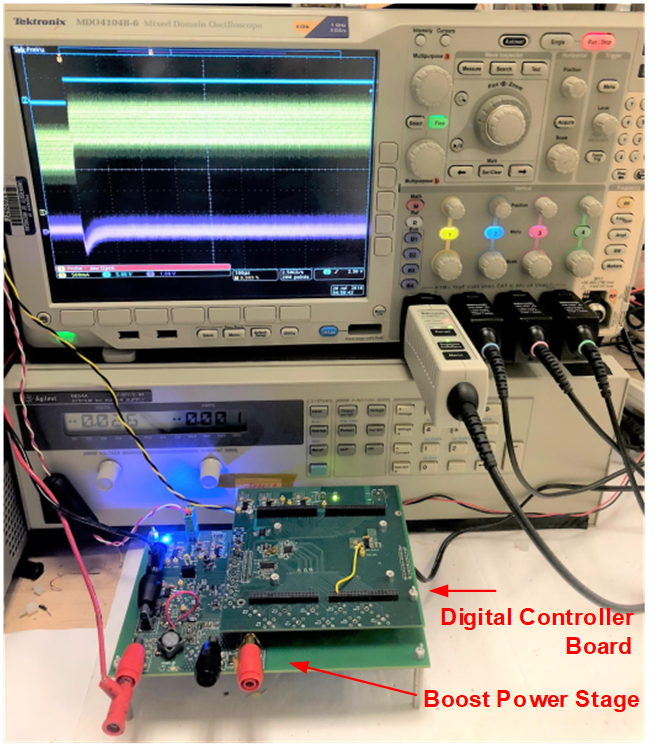}
    \caption{\label{fig:prototype_compel2019} Current-mode constant off-time boost converter and digital control hardware is under test.} 
\end{figure}
 \begin{figure}[htbp]
    \begin{center}
       \includegraphics[width=7cm]{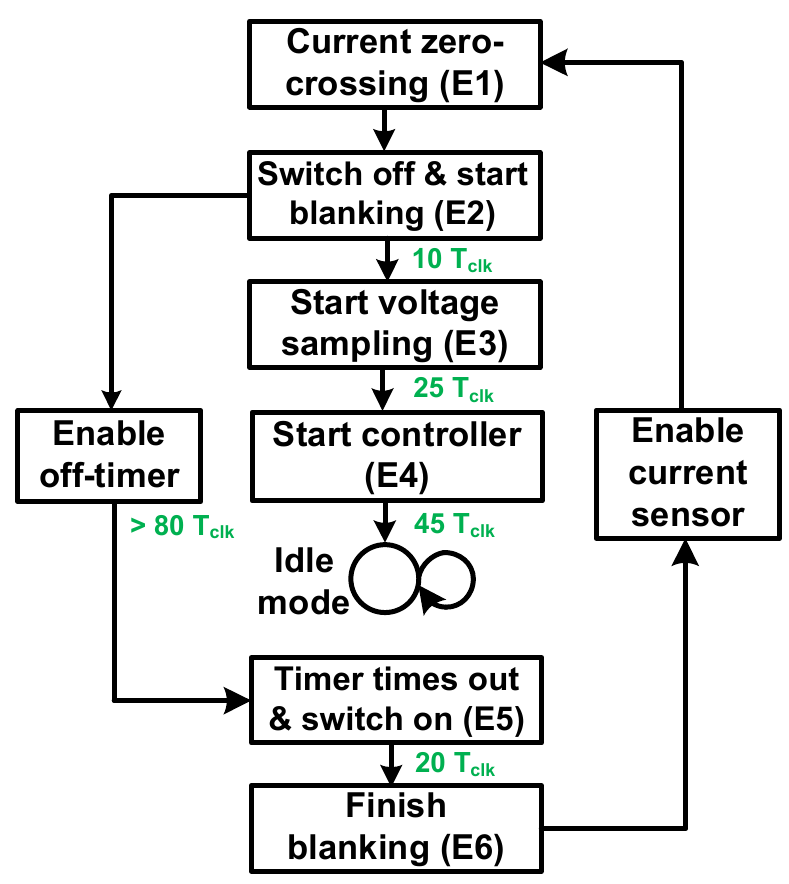}
    \end{center}
    \caption{\label{fig:flowchart} Switching\nobreakdash-synchronized sampling and control flowchart of a COT\nobreakdash-CM boost converter.}
\end{figure}

\begin{figure}[htbp]
    \begin{center}
       \includegraphics[width=\columnwidth]{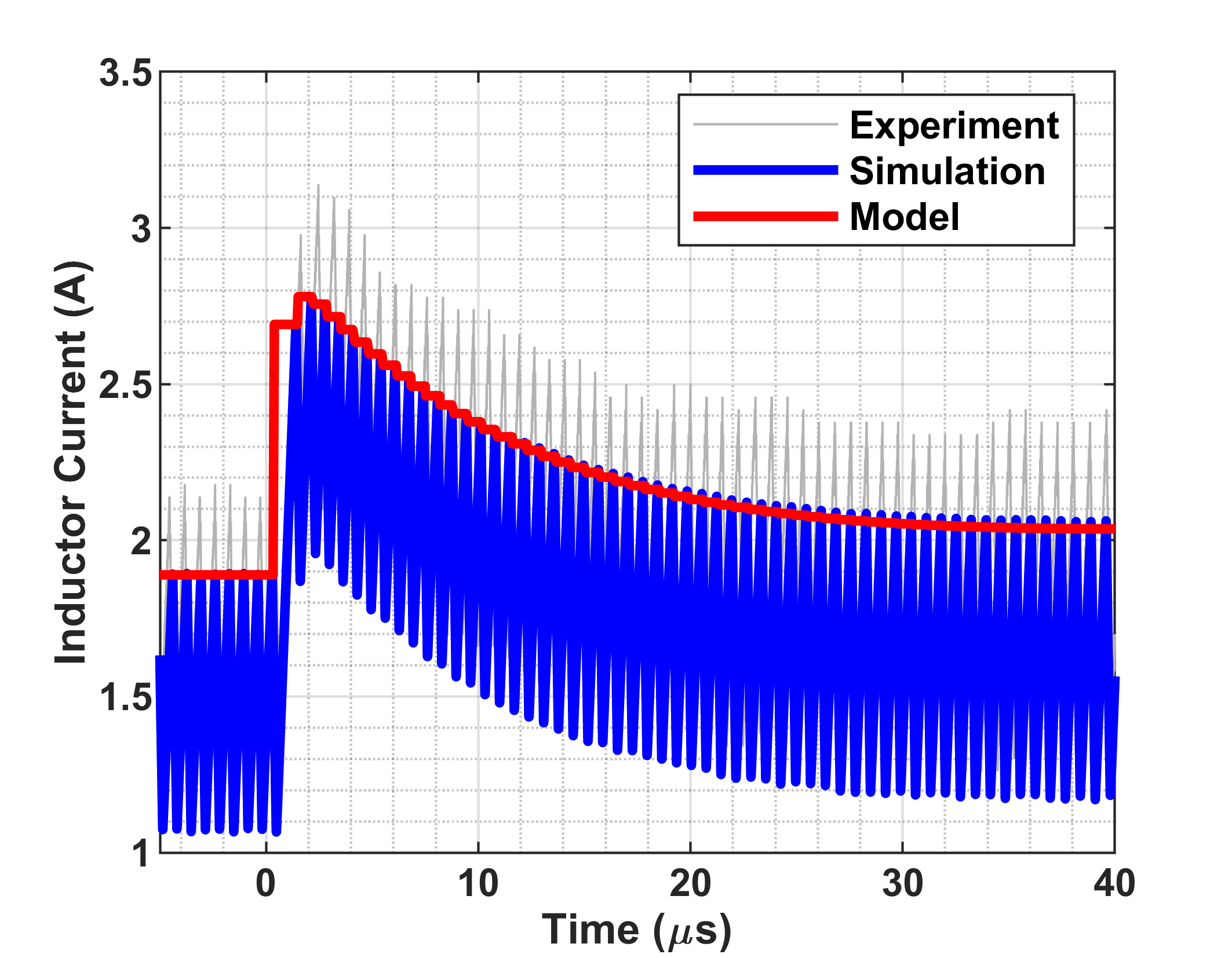}
    \end{center}
    \caption{ \label{fig:compinductor} Comparison of the inductor current $i_L$ waveform of a CM\nobreakdash-COT boost converter between the theory, simulation and experiment under output voltage step response.}
\end{figure}
\begin{figure}[htbp]
    \begin{center}
       \includegraphics[width=\columnwidth]{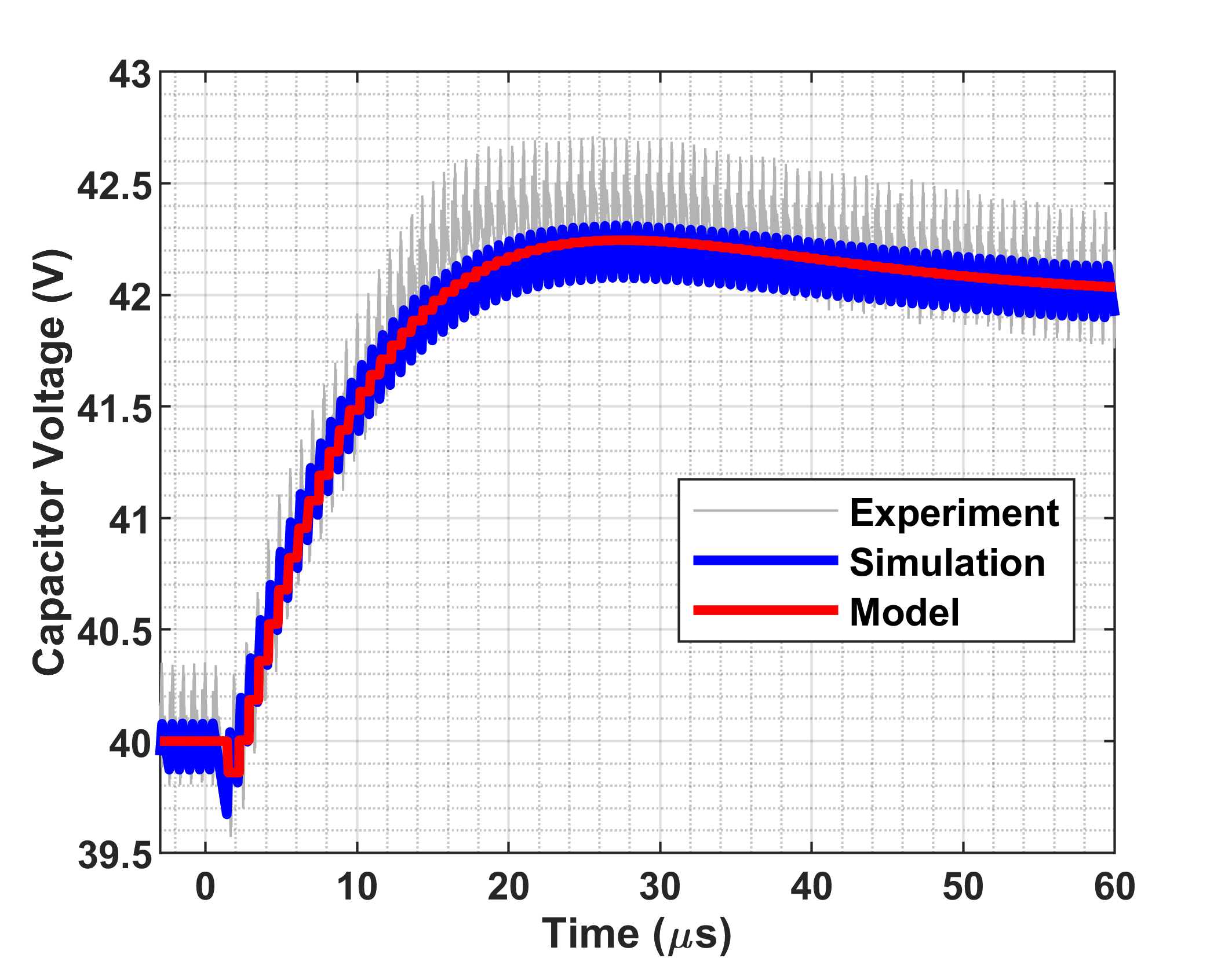}
    \end{center}
    \caption{\label{fig:compcapacitor} Comparison of the capacitor voltage $v_c$ waveform of a CM\nobreakdash-COT boost converter between the theory, simulation and experiment under output voltage step response.} 
\end{figure}
\begin{table}[tb]
    \caption{Design Parameters of the Constant Off\nobreakdash-Time Current\nobreakdash-Mode \\ Boost Converter Prototype} \label{tab:boosttable}
    \centering
    \vspace{5 pt}
    \begin{tabular}{cccccccccccccc}
        \toprule
    \textbf{Specifications/Items}&\textbf{Parameters}\\
    \midrule
    Peak Switching Frequency
    & 3 MHz\\
    \midrule
    Nominal Power
    & 16 W\\
    \midrule
    Input Voltage
    & 12 V\\
    \midrule
    Nominal Output Voltage
    & 40 V\\
    \midrule
    Off-Time
    & 200 ns\\
    \midrule
    $L$ / $C$
    &6.8 $\mu$H / 1 $\mu$F\\
    \midrule
    ADC / DAC
    & LTC2378-16 / MAX5184\\
    \midrule
    MOSFET / Diode
    & GS61004B / STPS1H100A\\
        \bottomrule
    \end{tabular}
\end{table}

Figs.\,\ref{fig:compinductor} and \ref{fig:compcapacitor} show good agreement between theory, simulation, and experiment. The theoretical current is offset from the experimental data by 10\%. This deviation is a result of the assumption in the theory that the converter is lossless; the prototype is instead 90\% efficient. The theoretical voltage matches the experimental data in steady\nobreakdash-state because of the integrator in the controller. 
The actual voltage ripple is approximately 200\,mV in Fig. \ref{fig:compcapacitor}.

The practical DVS task for LiDAR involves large\nobreakdash-signal voltage steps. 
The small\nobreakdash-signal model may cause large errors, as shown in Fig.\,\ref{fig:errorwpertubarion}, and not provide us with the desired transient response. 
To extend our small-signal approach to large\nobreakdash-signal voltage steps, we use \emph{gain scheduling}\cite{khalil2002nonlinear}. We select the output voltage $v_{\text{out}}$ as the \emph{scheduling variable} to parameterize the operating points of the COT\nobreakdash-CM boost converters.

To explain this technique, we first define the \emph{linearized region} of an operating point $u_e$ as a neighborhood of $u_e$ such that all dynamics occurring within this neighborhood can be well\nobreakdash-approximated by the linearized model at $u_e$. The basic principle of gain scheduling is dividing a large\nobreakdash-step control task into several small\nobreakdash-step sub\nobreakdash-tasks so that their linearized regions overlap with each other. We discretize the continuous operating\nobreakdash-point space of a COT\nobreakdash-CM boost converter into a discrete operating\nobreakdash-point set. At each critical operating point, a local controller is designed. This parameterized family of linear controllers can be easily stored as a look\nobreakdash-up table in the ROM of FPGA.
A supervisory controller first decomposes the voltage step task into multiple sub\nobreakdash-tasks. The local controller is activated and starts the sub\nobreakdash-task. The supervisory controller records and evaluates the voltage\nobreakdash-error series along the time. Once the error series is marked as ``settled\nobreakdash-down'', the supervisory controller governs the voltage references to the next voltage target and switches to the next local controller to trigger the next voltage step sub\nobreakdash-task.

Our digital implementation makes it quite easy for the the supervisory controller to cooperate with a family of local controllers. This is an important advantage over analog controllers.
The LiDAR transmitter, which is modeled as a resistive load in model (\ref{eqn:modifiedplantmodel}), is not necessarily a pure resistor. It is usually reasonable to assume that the dynamics of the LiDAR transmitter is so fast that its VI characteristics is instantaneous. Then given any electrical characteristics $i =
f(v)$, a new dynamic model can be constructed by replacing $R$ by $r = 1/\frac{df(v)}{dv}\Big |_{v=v_{\text{out}}}$.

The experimental large-signal staircase voltage steps from 20 V $\rightarrow$ 25 V $\rightarrow$ 30 V $\rightarrow$ 35 V $\rightarrow$ 40 V shown in Fig. \ref{DVSLiDAR} emulate a practical dynamic laser pulse peak power corresponding to 60~W $\rightarrow$ 80~W $\rightarrow$ 100~W $\rightarrow$ 115~W $\rightarrow$ 125~W. We use the gain scheduling from simulation for the hardware experiment. Each voltage step exhibits a rise time of approximately 5\,$\mu$s with small overshoot, which satisfies the dynamic performance requirements for state\nobreakdash-of\nobreakdash-the\nobreakdash-art LiDAR transmitter systems. A load step from 16\,W to 22.4\,W at 40\,V output voltage shown in Fig. \ref{LoadDisRejection} emulates a laser pulser repetition rate step from 700\,kHz to 1\,MHz. Under a large load disturbance, the maximum voltage deviation is 1\,V, which is within the 5\% capacitor bank discharge limit \cite{Diodes2018}.
 \begin{figure}[htbp]
    \begin{center}
       \includegraphics[width = \columnwidth]{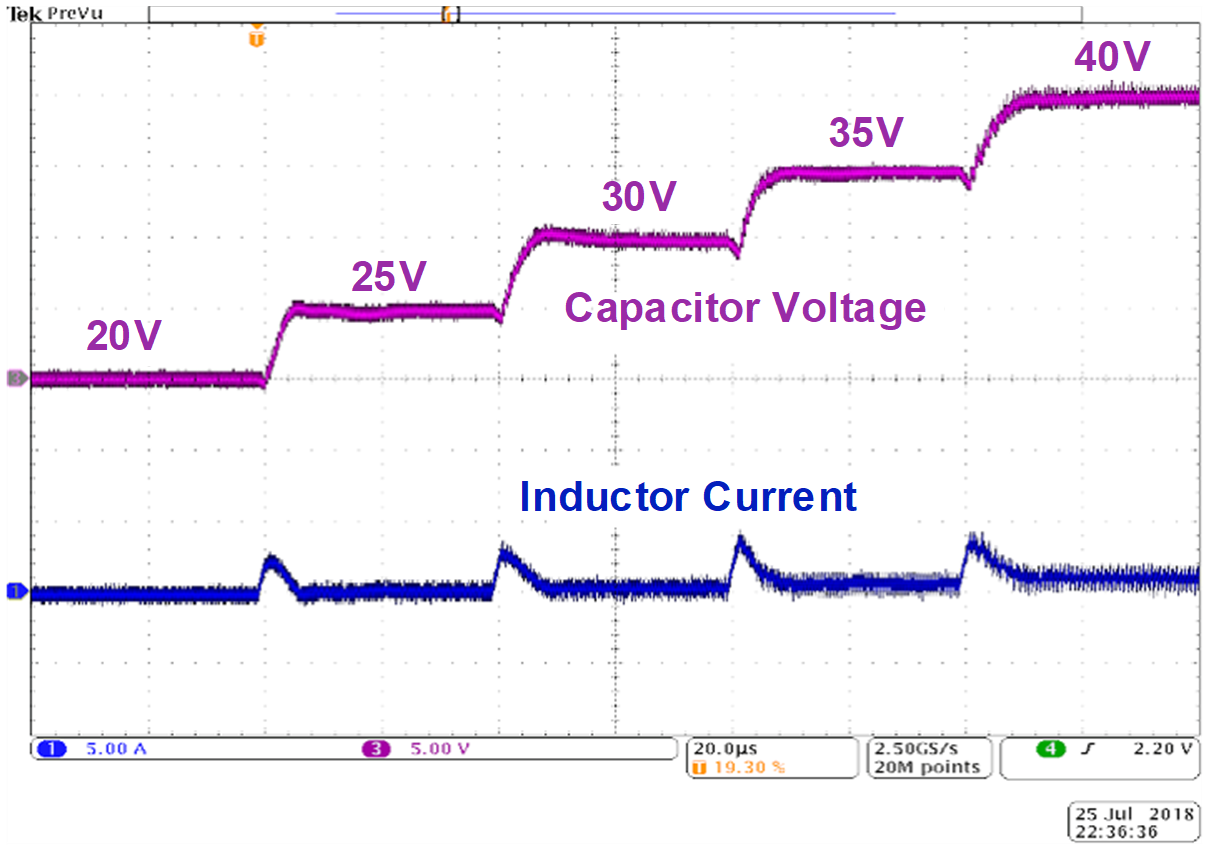}
    \end{center}
    \caption{\label{DVSLiDAR} Dynamic laser pulse energy scaling.} 
\end{figure}
\begin{figure}[htbp]
    \begin{center}
       \includegraphics[width=\columnwidth]{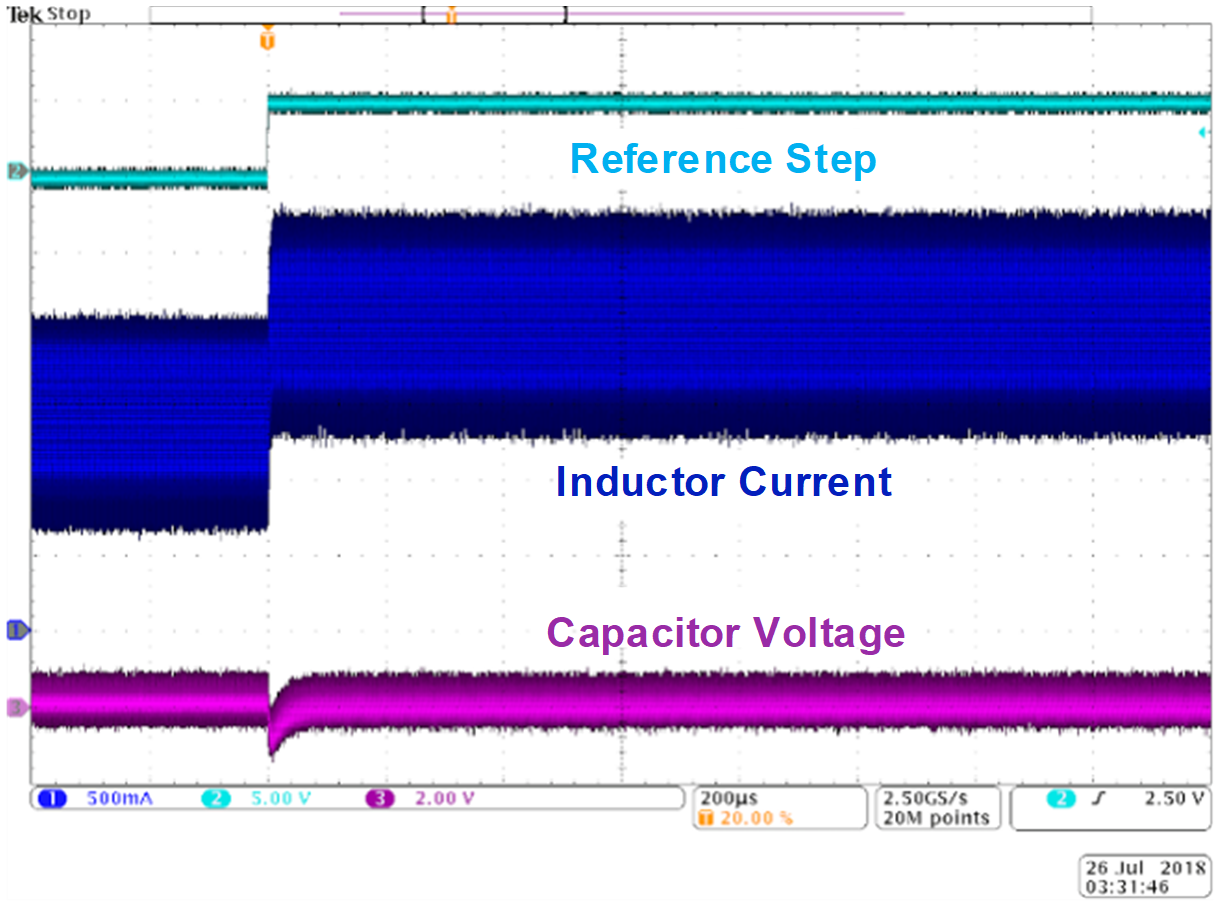}
    \end{center}
    \caption{\label{LoadDisRejection} Laser pulser repetition rate step.} 
\end{figure}
\section{Conclusion} \label{sec:conclusion_5s}
In this paper, we presented a new model and design methodology for performing switching cycle event\nobreakdash-driven digital control on variable\nobreakdash-frequency dc\nobreakdash-dc converters. We demonstrated an accurate model for a dc\nobreakdash-dc converter plant in a non\nobreakdash-periodic sampled state space.
We illustrated a method for designing a switching\nobreakdash-synchronized controller for a given plant. Closed\nobreakdash-loop system performance from an analytical model is verified through simulations and experiments. Dynamic voltage scaling as an application for dc\nobreakdash-dc converters using a switching\nobreakdash-synchronized controller can largely improve the energy efficiency of processors, memories, communications circuits, LiDAR power supplies, among others, by responding to the varying energy demand at an extremely fast speed and a very flexible way.

\section*{Acknowledgements}
This work was supported in part by the U.S. Department of Energy SunShot Initiative, under Award Number(s) DE\nobreakdash-EE\nobreakdash-0007549. 

{\setstretch{1}\vspace{\baselineskip}
\bibliographystyle{ieeetr}
\bibliography{library_fixed.bib}
}

\end{document}